\documentclass[12pt]{amsart}
\usepackage{amsmath,amssymb, amscd}
\usepackage[all]{xy}
\usepackage{color}
\usepackage[dvipdfmx]{graphicx}

\makeatletter
\@addtoreset{equation}{section}
\makeatother
\usepackage{enumerate}

\def\ii{{\sqrt{-1}}}

\def\cE{{\mathcal{E}}}

\def\fC{{\mathfrak{C}}}

\def\cZSC{{\mathcal{Z}}_{\mathrm{SC}}}
\def\cZBCC#1{{\mathcal{Z}_{\mathrm{BCC}}^{#1}}}
\def\BB{{\mathbb{B}}}

\def\cR{{\mathcal{R}}}

\def\cM{{\mathcal{M}}}
\def\cF{{\mathcal{F}}}

\def\ee{\mathrm e}

\def\cF{{\mathcal{F}}}

\def\SL{\mathrm {SL}}

\def\PSL{\mathrm {PSL}}

\def\AA{{\mathbb A}}

\def\CC{{\mathbb C}}
\def\ZZ{{\mathbb Z}}

\def\QQ{{\mathbb Q}}
\def\RR{{\mathbb R}}
\def\EE{{\mathbb E}}
\def\HH{{\mathbb H}}

\def\cF{{\mathcal{F}}}

\def\UU{{\mathrm{U}}}

\def\SC{{\mathrm{SC}}}
\def\BCC{{\mathrm{BCC}}}

\def\LA{\langle}
\def\RA{\rangle}

\def\displace{\mathfrak{u}}

\theoremstyle{plain}
\newtheorem{theorem}{Theorem}[section]
\newtheorem{proposition}[theorem]{Proposition}
\newtheorem{lemma}[theorem]{Lemma}

\newtheorem{remark}[theorem]{Remark}
\def\dfrac#1#2{{\displaystyle\frac{#1}{#2}}}

\def\by#1{{\rm{#1},}}

\begin{document}

\title[A novel discrete investigation on screw dislocations in 
BCC lattices]
{A novel discrete investigation on  screw dislocations in the BCC 
crystal lattices}

%about the paper that should go on the front page should be
%placed here. General acknowledgments should be placed at the end of the paper%.}
%\titlerunning{Stress energy of a screw dislocation in BCC lattice}

%\titlerunning{An algebraic description of screw dislocations}        % if too long for running head

\author{
Shigeki Matsutani 
}

\maketitle

\begin{abstract}
In this paper, we proposed a novel method using the elementary
number theory to investigate the discrete nature of 
the screw dislocations in crystal lattices, simple cubic (SC) lattice
and body centered cubic (BCC) lattice, by developing
the algebraic description of the dislocations in the 
previous report (Hamada, Matsutani, Nakagawa, Saeki, Uesaka, 
Pacific J. Math.~for Industry {\bf{10}} (2018), 3).
Using the method, we showed that the stress energy of the 
screw dislocations in the BCC lattice and the SC lattice are 
naturally described; the energy of the BCC lattice was expressed by
the truncated Epstein-Hurwitz zeta function of the Eisenstein integers,
whereas that of SC lattice is associated with
the truncated Epstein-Hurwitz zeta function of the Gauss integers.


MSC2020: {
08A99; %None of the above, but in this section, Algebraic structures
55R05; %Fiber spaces
20H15; % Other geometric groups, including crystallographic groups
11R60: %Cyclotomic function fields (class groups, Bernoulli objects, etc.)
34M35; %Singularities, monodromy, local behavior of solutions, normal forms
82D25; %Crystals {For crystallographic group theory, see 20H15}
74E15; % Crystalline structure (Mechanics of deformable solids)
82-10; %Mathematical modeling or simulation for problems pertaining to statistical mechanics
74-10; %Mathematical modeling or simulation for problems pertaining to statistical mechanics
}
\end{abstract}

\section{Introduction}
Since the dislocations in crystal lattices have effects on the 
 properties of the materials, i.e., elasticity, plasticity and fracture,
the screw dislocations have been studied from several viewpoints
\cite{AHL, HB,N}.
Recently the progress of technology in material
industry, especially steel industry,
requires much higher spec of the 
properties of material than those decades before
and requires to control the production processes of materials 
more highly
from various viewpoints.
It implies that in a next couple of decades, it will
be  necessary to control the dislocations much more precisely
than current quality.
The rapid development of technology also influences the experimental
equipments and thus recently 
we can directly observe micro-scopic and meso-scopic
features of materials even in crystal scale \cite{ISCKI, TH1, TH2};
the observation scheme could meet such expectations.
However there do not exist proper tools to represent such phenomena
in discrete nature in complex system;
the tools must be mathematical tools,
 which might be quite different from
the current  approaches.

In order to prepare for the drastic change in material science from the 
viewpoint of mathematical science,
we have had serial conferences for these five years in which mathematicians 
and material scientists including researchers in 
steel industry have discussed to provide 
the novel mathematical tools for next material science 
(see Acknowledgments); we provided a novel tool to 
describe the discrete nature in dislocation in terms of algebraic language
in the previous paper \cite{HMNSU}. In this paper, we develop the
previous result to describe the symmetry in the discrete nature of the 
screw dislocations well in terms of elementary number theory.
Using the elementary number theory, we focus on 
the expression of the difference between 
the screw dislocations in the simple cubic (SC) lattices and the
body centered cubic (BCC) lattices.
Though some of them are represented by
other methods, the number theoretic approach turns out to be
a good and natural tool for the description of the discrete
 systems, which will be the basic tool to investigate much more
complex systems.

\bigskip

Though the origin of the screw dislocations is a discrete nature of crystals,
the dislocations have been studied in the continuum picture
because 1) there was no proper method to describe their discrete structure
and 2) the continuum picture is appropriate for the behavior of macro-scale
of the dislocations.
In the geometrical description of the dislocations 
as a continuum picture \cite{Mer, KE, KF, LA}, 
which Kondo and Amari started to investigate \cite{Kon,A1,A2},  
the global behavior of the dislocations is expressed well.
Even in continuum picture of the dislocations 
including phenomenological models,
there are so many crucial mathematical problems which are
effective for the material science, e.g, 
\cite{CS, ENA, NN, SG}.

However as mentioned above, we cannot avoid to understand micro- and meso-scopic feature of dislocations and in order to understand them, mathematics also
plays important roles.
Since the positions of atoms in the crystal in the micro-scopic scale 
are fluctuated,
the micro-scopic properties of the 
dislocation have been investigated 
 by means of the molecular dynamics or molecular mechanics
in classical level and in the level of the first principle, e.g, 
\cite{Cl, GD, IKY}. 
It is a crucial problem, in mathematics, how we introduce links to
consider the topological properties for given position of atoms
in our euclidean space.

We are concerned with the properties in 
the meso-scopic scale, which cannot be represented by the continuum picture 
neither by the molecular mechanics nor the first principle approaches.
One of our purposes in this paper is to 
investigate the dependence of the dislocations on the type of crystals
mathematically.
Recently Ponsiglione \cite{P} and 
Alicandro, Cicalese and Ponsiglione \cite{ACP}
investigated the behavior of dislocations in 
the meso-scopic scale in the framework of $\Gamma$-convergence.
Hudson and Ortner \cite{HO} and Braun, Buze, and Ortner \cite{BBO}
considered the discrete picture of dislocations.
Ariza and Ortiz \cite{AO}, 
Ramasubramaniam, Ariza and Ortiz \cite{ARO},
and Ariza, Tellechea, Menguiano and Ortiz \cite{ATMO} studied the 
discrete nature of the dislocations in terms of modern mathematics, i.e.,
homology theory, graph theory, group theory and so on.
Especially Ariza and Ortiz \cite{AO} and 
Hudson and Ortner \cite{HO}
provided geometrical methods to reveal the 
discrete nature of dislocations and 
studied the core energy of the dislocations of the 
BCC lattice.

We recall that 
the crystal lattices have high symmetries such as translational and 
rotational symmetries governed by the crystal groups, which are
studied in the framework of crystallography.
These symmetries are described well in terms of algebraic language and
algebraic tools 
in wider meaning \cite{CS,S}.
Representation of finite groups is representation of their group rings and
modules in the module theory.
The lattice $\ZZ^n$ in Euclidean space $\EE^n$
has been studied in the number theory, which is 
known as Minkowski arithmetic geometry and related to the quadratic fields
and the harmonic analysis such as the Epstein zeta function \cite{Tre}.
The two dimensional lattice, $\ZZ+\ZZ\tau (\subset \CC)$, 
($\tau \in \HH:=\{x+y\ii \in \CC \ |\ y>0\}$),
 associated with the elliptic curves
has been studied well in the study of modular forms \cite{Knapp,IR}.
The action of $\SL(2,\ZZ)$ on the lattice and its subgroup show the 
symmetry of the lattice $\ZZ+\ZZ\tau$. 
When $\tau=\ii$ and $\tau=\omega_6$ (or $\omega_3$)
for $\omega_p=\ee^{2\pi\ii/p}$,
they are known as the Gauss integers and  the Eisenstein integers respectively
\cite{IR,Tri}. They have been studied well in the framework of 
the algebra and the algebraic number theory.

It is emphasized that the crystal lattices even with defects and their
 interfaces still have higher symmetries.
They should be regarded as a kind of
symmetry breaking of the group \cite{W}.
It means that they are not stable for the crystal group in general but
are stable for its subgroup, at least, approximately,
and should be described by algebraic theory cooperated with
analytic and geometric theories.

The interfaces of two crystal lattices are described well by
the quadratic fields in the elementary number theory 
\cite{ISCKI, KAKT}.
Thus even for the dislocations, we should express their symmetry properly.
In the previous report \cite{HMNSU} with
Hamada, Nakagawa, Saeki and Uesaka, 
we focused on the fiber structure of the screw dislocations as
an essential of the screw dislocations.
The bundle map in the Cartesian square realizes the screw dislocations
in the SC  and the BCC lattices induced from the continuum
picture.
The fiber structure shows the translational symmetry of the fiber direction,
which is the survived symmetry in the these crystal lattices even if 
the screw dislocation exists.
On the other hand, the vertical to the fiber direction (the direction of 
Burgers vector) there are other symmetries which are induced from the 
 crystal group for the perfect crystals, i.e., the two-dimensional crystal 
lattices. 
Though we did not argue the analytic properties in \cite{HMNSU},
when we consider the minimal point of the configuration of the 
atoms, their initial configuration should be indexed by 
natural indices reflecting the symmetry of the dislocation.

In this paper, we extend the method in previous report to 
express the difference between the screw dislocations in
 the SC and the BCC lattices
algebraically.
We propose a novel method to
investigate the algebraic nature of 
the screw dislocation in crystal lattices, the SC and the BCC
lattices, using the elementary number theory;
the Gauss integers $\ZZ[\ii]$ and  the Eisenstein integers 
$\ZZ[\omega_3]=\ZZ[\omega_6]$ correspond to the vertical two-dimensional
lattices for the screw dislocations of the SC and the BCC lattices.
Our method shows the natural indices of configurations of atoms,
which must be useful even when we consider their analytic properties.
For examples, as in Remarks \ref{rmk:AlgSC} and \ref{rmk:AlgBCC},
and Lemmas \ref{lm:cyl} and \ref{lm:4.2},
 the ring of integers $\ZZ[\tau]$ of the cyclotomic field
$\QQ[\tau]$ show the algebraic properties in these lattices.
Especially, we investigate the symmetry of the two-dimensional crystal lattice
in terms of $\ZZ[\tau]$ to show
the critical relations between the energy of the dislocations and 
the Epstein-Hurwitz zeta 
functions, in the SC and the BCC lattices
as we show in 
Theorems  \ref{thm:energySC} and  \ref{thm:energyBCC}.
The number theoretic approach shows the symmetry of these systems well.

This paper is organized as follows.
Section~\ref{section2} and \ref{section3}  
review the previous report \cite{HMNSU}.
In Section~\ref{section2}, we show the screw dislocation in the 
continuum picture.
Section~\ref{section3} reviews the results of the SC lattice case in 
\cite{HMNSU} in terms of Gauss integers $\ZZ[\ii]$.
In Section~\ref{section4}, after we also show the configuration of the screw
dislocation in the BCC lattice in terms of the 
Eisenstein integers $\ZZ[\omega_6]$
following \cite{HMNSU},
we provide the algebraic expression of the stress energy 
of the screw dislocation in the BCC lattice, which is our main result in 
this paper.
In Section 5, we discuss these results and in Section 6,
we summarize our results.

\section{Screw Dislocations in Continuum Picture}\label{section2}

In this section, we  review the previous report \cite{HMNSU} 
and show the algebraic expression of the 
screw dislocations in continuum picture.

\subsection{Notations and Conventions}

Since the translational symmetry is crucial in physics
\cite{W},
in this paper, we distinguish the euclidean space $\EE$ from
the real vector space $\RR$:
 we regard that $\RR$ is a vector space,
whereas $\EE$ is the space consisting of the position vectors with
translational symmetry,
though both $\EE^n$ and $\RR^n$ are topological spaces with the ordinary
euclidean topology. 
Similarly we distinguish the set of the complex position vector,
the affine space $\EE_\CC$, from the complex vector space $\CC$.
We basically identify the $2$-dimensional euclidean space
$\EE^2$ with $\EE_\CC$, and $\RR^2$ with $\CC$. 
The group $U(1)$ naturally acts on the circle $S^1$.
$\ZZ$ and $\QQ$ are the sets of the rational integers and
the rational numbers respectively.
For a fiber bundle $\cF \rightarrow \cM$ over a base space
$\cM$, the set of continuous sections $f:\cM \to \cF$
is denoted by $\Gamma(\cM, \cF)$.

In this paper, for $\delta =(\delta_1, \delta_2, \delta_3)\in \EE^3$,
 we consider an embedding $\iota_\delta$ of the vector space 
$\RR^3=\CC\times \RR$ into
$\EE^3=\EE_\CC \times \EE$ by
$$
\iota_\delta : \RR^3 \hookrightarrow \EE^3,
\quad (x \mapsto x+\delta), \quad
\mbox{or}\quad
$$
$$
\iota_\delta : \CC\times \RR \hookrightarrow \EE_\CC \times \EE,
\quad ((x_1+\ii x_2, x_3) \mapsto (x_1+\ii x_2+\delta_\CC, x_3+\delta_3 ),
$$
where $\delta_\CC:=\delta_1+\ii \delta_2$.

Further we employ some conventions listed up in Appendix.

\subsection{Exact Sequence and Sequence of Maps}

We consider the exact sequence of groups (see \cite{B}),
\begin{equation}
\xymatrix@!C=50pt{
0 \ar[r] & \ZZ\ \ar[r]^-i & \RR
        \ar[r]^-{\exp 2\pi \ii} & \UU(1) \ar[r]^-{} & 1,}
\label{eq1}
\end{equation}
which is essential in this paper.
$\ZZ$ and $\RR$ are additive groups, $\UU(1)$ is a
multiplicative group,
$i(n)=n\in \RR$ for $n \in \ZZ$, and
$(\exp 2\pi \ii)(x) = \exp (2\pi \ii x)$ for $x \in \RR$.

In our description of the screw dislocations, we fix the third axis as
the direction of the Burgers vector.
For $\delta_3$ in $\delta=(\delta_1, \delta_2, \delta_3) \in \EE^3$ and 
a certain positive number $d > 0$ which
is given as $d=a$ in Section \ref{section3}
and $d = \sqrt{3}a/2$ in Section \ref{section4}, 
we define
the shifted maps,
\begin{eqnarray*}
 & \widetilde i_{d,\delta}:\RR \to \EE, &
(x \mapsto d\cdot x +\delta_3), \\
 & i_{d,\delta}:\UU(1) \to S^1, & (\exp({\ii\theta}) \mapsto
\exp{\ii(\theta +2\pi\delta_3/d)})
\end{eqnarray*}
satisfying the commutative diagram,
$$
\xymatrix@!C=50pt{
 & & \EE\ \ar[r]^{\psi_d} & S^1\\
0 \ar[r] & \ZZ  \ar[r]^-{i} \ar[ru]^-{\varphi_{\delta}}
& \RR \ar[u]_{\widetilde{i}_{d,\delta}}
        \ar[r]^{\exp 2\pi\ii} & \UU(1) \ar[u]_{i_{d,\delta}}
\ar[r]^-{} & 1, }
$$
where $\psi_d(y) = \exp(2\pi\ii y/d)$, $y \in \EE$, and
$\varphi_{\delta} = \widetilde{i}_{d,\delta}
\circ i$. 
It means that we have the sequence of maps
\begin{equation}
\xymatrix@!C{
\ZZ\ \ar[r]^-{\varphi_{\delta}} & \EE
        \ar[r]^-{\psi_d}  & \ S^1,}
\label{eq:seq_maps}
\end{equation}
where
\begin{equation} \label{eq:varphi}
\varphi_{\delta}(\ZZ) =
\psi_d^{-1}(\exp(2\pi\ii\delta_3/d)).
\end{equation}

\subsection{Fiber Structures of Crystals in Continuum Picture}
Let us consider   some trivial bundles over $\EE_\CC$;
$\ZZ$-bundle $\pi_{\ZZ} : \ZZ_{\EE_\CC} \to \EE_\CC$,
$\EE$-bundle $\pi_{\EE} : \EE_{\EE_\CC} \to \EE_\CC$
and $S^1$-bundle $\pi_{S^1} : S^1_{\EE_\CC} \to \EE_\CC$.
The sequence of maps (\ref{eq:seq_maps}) induces
the sequence of bundle maps $\widehat\varphi_{\delta}$ and
$\widehat{\psi}_d$,
\begin{equation}
\xymatrix@!C{
 \ZZ_{\EE_\CC} \ar[r]^-{\widehat\varphi_{\delta}} &
 \EE_{\EE_\CC} \ar[r]^-{\widehat {\psi}_d} & S^1_{\EE_\CC}}.
\label{eq:01}
\end{equation}
It is obvious that
$\EE_{\EE_\CC}$ is identified with our three-dimensional euclidean
space $\EE^3 = \EE \times \EE_\CC$ whereas
$\ZZ_{\EE_\CC} = \ZZ\times \EE_\CC$ is a covering space of $\EE_\CC$.
$\ZZ_{\EE_\CC}$ expresses the geometrical objects which consist of
parallel sheets over $\EE_\CC$.

\bigskip

Let $\delta=(\delta_\CC=\delta_1+\ii \delta_2, \delta_3)\in \EE^3$.
We consider
the embedding
$$
\iota_\delta: \CC \hookrightarrow \EE_\CC, \quad
(x+\ii y \mapsto x+\ii y+\delta_\CC).
$$
This $\iota_\delta$ plays important roles in the following sections
by restricting the domain $\CC$ into its discrete sets $\ZZ[\ii]$ and 
$\ZZ[\omega_6]$, and thus
we set $z = x+\ii y+\delta_\CC \in \EE_\CC$.

Let us consider the image of the bundle map $\widehat\varphi_{\delta}$.
For $\gamma_\delta=\exp(2\pi\ii\delta_3/d) \in S^1$, we define 
 the global constant section $\displace_\delta\in
\Gamma(\EE_\CC, S^1_{\EE_\CC})$ of $S^1_{\EE_\CC}$ by
$$
\displace_\delta(z) =\gamma_\delta \in S^1_{\EE_\CC}|_z =
 S^1 \times \EE_\CC|_z,
$$
for $z \in \EE_\CC$. The following lemma is naturally obtained:

\begin{lemma}\label{lm:2.1}
For $\gamma_\delta = \exp(2\pi\ii\delta_3/d)$,
we have
$$
\ZZ_{\EE_\CC, \delta}=\widehat\varphi_{\delta}(\ZZ_{\EE_\CC}),
$$
where
$$
\ZZ_{\EE_\CC, \delta}:=
{\widehat{\psi}_d^{-1}}
\left(
\displace_{\delta}(\EE_\CC)\right) \subset \EE^3=\EE_{\EE_\CC}.
$$
\end{lemma}
Here we note that $\widehat\varphi_{\delta} (\ZZ_{\EE_\CC}) $ 
is the system consisting of parallel equi-interval sheets
realized in the three-euclidean space $\EE^3=\EE_{\EE_\CC}$.

\subsection{Single Screw Dislocation in Continuum Picture}\label{subsec2.2}

For $z_0 \in \EE_\CC$, let us consider the non-trivial bundles
$\EE_{\EE_\CC\setminus \{z_0\}}$ and
$S^1_{\EE_\CC\setminus \{z_0\}}$
over $\EE_\CC \setminus \{z_0\}$.
In other words,
we consider the section
$\displace_{z_0, \delta} \in
\Gamma(\EE_\CC\setminus \{z_0\}, S^1_{\EE_\CC\setminus \{z_0\}})$ 
defined by
\begin{equation}
\displace_{z_0, \delta}(z) = \gamma_\delta\frac{z-z_0}{|z-z_0|}
\mbox{ for } z \in \EE_\CC \setminus \{z_0\},
\label{eq:sigmaz0}
\end{equation}
and a natural universal covering of $\EE_\CC\setminus \{z_0\}$,
\begin{equation}
\ZZ_{\EE_\CC \setminus \{z_0\}, \delta}:=
{\widehat{\psi}_d^{-1}}
\left(\displace_{z_0, \delta}(\EE_\CC\setminus\{z_0\})\right) \subset
\EE_{\EE_\CC\setminus \{z_0\}} \subset \EE^3
\label{eq:Z_EE2z0}
\end{equation}
by  letting the restriction $\pi_{z_0, \delta}=
\pi_{\EE}|_{\ZZ_{\EE_\CC \setminus \{z_0\}, \delta}}$, i.e.,
$\pi_{z_0, \delta}: \ZZ_{\EE_\CC \setminus \{z_0\}, \delta} \to
\EE_\CC\setminus \{z_0\}$.
In this paper,
we call this covering {\it{a screw dislocation}} in a continuum picture
which is realized as a subset of $\EE^3$ following \cite{AHL, HB,N};
In these textbooks \cite{AHL, HB,N}, 
$\ZZ_{\EE_\CC \setminus \{z_0\}, \delta}$ is given by geometrical
consideration as a screw dislocation,
which is mentioned in Remark \ref{rk:2.2},
whereas it should be noted that
our construction of $\ZZ_{\EE_\CC \setminus \{z_0\}, \delta}$ is purely
algebraic.

As in \cite{HMNSU}, it is not difficult to extend this expression
of the single screw dislocation to 
one of multi-screw dislocations.

\begin{remark}\label{rk:2.1}
For the simply connected neighborhood $U_p \subset \EE_\CC \setminus \{z_0\}$
of a point $p$ of $\EE_\CC \setminus \{z_0\}$,
$$
\pi_{\EE}^{-1}U_p \cong \ZZ \times U_p,
$$
as a covering space of $U_p$.
\end{remark}

\begin{remark}\label{rk:2.2}
{\rm{
$\ZZ_{\EE_\CC \setminus \{z_0\}, \delta}$ can be obtained by the following the
operation on the trivial covering 
$\ZZ\times (\EE_\CC \setminus \{z_0\})$ with the embedding 
$\iota_{\EE}:\ZZ\times (\EE_\CC \setminus \{z_0\})\hookrightarrow \EE^3$,
such that $\pi_{\EE}: \iota_{\EE}(\ZZ\times (\EE_\CC \setminus \{z_0\})) \to 
\EE_\CC \setminus \{z_0\}$.
We regard it as the set of sheets indexed by the integers $n$.
The third position component of the $n$-th sheet is given by $n d+\delta_3$.
Let us consider a half line $L :=\{x+\ii y_0 \ | \ x \ge x_0\}$ for 
$z_0 = x_0 +\ii y_0$ 
and $\EE_\CC\setminus L$ as a simply connected open set of 
$\EE_\CC \setminus \{z_0\}$.
First we cut 
$\iota_{\EE}(\ZZ\times (\EE_\CC \setminus \{z_0\}))$
at the inverse $\pi_{\EE}^{-1}(L) \subset \EE^3$.
In other words, we consider 
$\pi_{\EE}^{-1}(\EE_\CC \setminus L)$ noting Remark \ref{rk:2.1}.
We deform the $n$-th sheet in $\EE^3$
such that the third component is 
given by $n d+\delta_3 + 
\displaystyle{\frac{d}{2\pi} \mathrm{arg}\frac{z-z_0}{|z-z_0|}}$.
After then,
 we connect the $n$-th sheet to the $(n+1)$-th sheet at
the place $\pi_{\EE}^{-1}(L)$. Then we obtain 
$\ZZ_{\EE_\CC \setminus \{z_0\},\delta}$ in (\ref{eq:Z_EE2z0}).
It means that this is a construction of
$\ZZ_{\EE_\CC \setminus \{z_0\},\delta}$ as a discontinuous 
deformation of
$\ZZ_{\EE_\CC,\delta}$, which 
is the standard geometrical description of the dislocation
\cite{AHL, HB,N}.
}}
\end{remark}

\section{Screw Dislocation in Simple Cubic Lattice}\label{section3}

In this section, we show the algebraic description of the 
screw dislocation in the SC lattice
and its stress energy in terms of 
the Gauss integers $\ZZ[\ii]\subset \CC$ (see Appendix).

\subsection{SC Lattice as Covering Space of $\ZZ[\ii]$}
\label{section3.1}
For the SC lattice in the three euclidean space $\EE^3$,
$$
\ZZ_{\SC,\delta}:=
\{(\ell_1a , \ell_2a, \ell_3 a) +\delta \ | \ 
\ell_1, \ell_2, \ell_3 \in \ZZ\},
$$
where $\delta = (\delta_1, \delta_2, \delta_3) \in \EE^3$,
and $a$ is the lattice length $(a>0)$, we find its
fiber structure as in the previous section. 
Let 
$
\cZSC:=\{n_1 a + n_2 a\ii \ |\ n_1, n_2 \in \ZZ\}\subset \CC,
$
which can be expressed by the Gauss integers $\ZZ[\ii] = \ZZ +\ZZ \ii$,
$$
\cZSC=\ZZ[\ii]a\subset \CC.
$$

For $\delta = (\delta_1, \delta_2, \delta_3) \in \EE^3$,
we define the embedding,
\begin{equation}
\iota^\SC_{\delta}: \cZSC \to 
\cZSC + \delta_\CC \subset \EE_\CC,
\label{eq:3.1}
\end{equation}
 where,
$\delta_\CC = (\delta_1 + \delta_2 \ii) \in \EE_\CC$.
The embedding $\iota^\SC_{\delta}$ induces the
bundle map $\widehat\iota^\SC_{\delta}$.

Using $\displace_{\delta}(z)$ in Lemma \ref{lm:2.1} of 
$\gamma_\delta := \exp(2\pi\ii\delta_3/a)$ for the position
$\delta\in \EE^3$,
 we reconstruct 
 the SC lattice $\ZZ_{\SC,\delta}$ by
$$
\ZZ_{\cZSC,\delta}= 
{\widehat{\psi}_a^{-1}}
\left(
\displace_{\delta}(\iota^\SC_{\delta}(\cZSC))\right), 
$$
which is realized in $\EE^3$,
$\ZZ_{\SC,\delta}= \ZZ_{\cZSC,\delta} \subset \EE^3$.
Here we set $d=a$ in $\psi_d$ in the previous section.

\begin{remark}\label{rmk:AlgSC}
{\rm{
Corresponding to  
Lemmas \ref{lm:cyl} and \ref{lm:4.2}, and Remark \ref{rmk:AlgBCC} for the
BCC lattice case, we have the formula in $\ZZ[\ii]$,
\begin{equation}
\sum_{\ell=0}^3 (\ii)^\ell = 0,
\label{eq:AlgSC}
\end{equation}
which is known as the cyclotomic symmetry of $\ZZ[\ii]$
or the cyclic group $\fC_4$ action of the order 4  on $\ZZ[\ii]$.
This relation makes the formula
(\ref{eq15a}) simply described and
connected with the Epstein-Hurwitz zeta function
as in Theorem \ref{thm:energySC}.
}}
\end{remark}

\subsection{Graph related to $\ZZ_{\cZSC,\delta}$}\label{subsection3.2}
We introduce the infinite graph $G^\SC_\delta$ whose nodes are given by 
$\ZZ_{\cZSC,\delta} \cong \ZZ^3$.

We consider the edges among the nodes in $G^\SC_\delta$.
As $G^\SC_\delta$ is parameterized by $\ZZ^3$, we consider the
edges
\begin{equation}
\begin{split}
&[(n_1, n_2, n_3), (n_1\pm1, n_2, n_3)],
[(n_1, n_2, n_3), (n_1, n_2\pm1, n_3)],
[(n_1, n_2, n_3), (n_1, n_2, n_3\pm1)],\\
&[(n_1, n_2, n_3), (n_1, n_2\pm1, n_3\pm1)],
[(n_1, n_2, n_3), (n_1\pm1, n_2, n_3\pm1)],\\
&[(n_1, n_2, n_3), (n_1\pm1, n_2\pm1, n_3)]
\end{split}
\label{eq:graphSC}
\end{equation}
for every point $(n_1, n_2, n_3)\in \ZZ^3\cong \ZZ_{\cZSC,\delta}$.
The first and the second components correspond to the horizontal directions
whereas the third one does to the vertical direction.

\subsection{Dislocation in SC Lattice as Covering Space of $\ZZ[\ii]$}

A screw dislocation in the simple cubic lattice appears along
the $(0,0,1)$-direction \cite{N} 
up to automorphisms of the SC lattice.
The Burgers vector is parallel to
the $(0, 0, 1)$-direction.

Using the fibering structure of $\EE_\CC\setminus\{z_0\}$,
we can describe a single screw dislocation in the SC lattice
as in \cite{HMNSU}.

For $\delta=(\delta_1, \delta_2, \delta_3) \in \EE^3$,
we also let $\gamma_{\delta}=\exp(2\pi \ii \delta_3/a)\in S^1$ and $\delta_\CC
= (\delta_1 + \delta_2 \ii)$.
Using (\ref{eq:3.1}),
let us define the section $\displace_{z_0, \delta}^\SC
\in \Gamma(\cZSC, S^1_{\cZSC})$ by 
$$
\displace_{z_0, \delta}^\SC := 
\iota^{\SC *}_{\delta_\CC}\displace_{z_0, \delta}=
\displace_{z_0, \delta}\circ
\iota^{\SC}_{\delta_\CC},
$$
$$
     \displace_{z_0,\delta}^\SC(n a)
          = \left(\gamma_\delta\frac{na + \delta_\CC - z_0}
                       {|na  + \delta_\CC - z_0|} \right),
     \quad na \in \cZSC=\ZZ[\ii]a.
$$

Using this $\displace_{z_0,\delta}^\SC$, we define its screw dislocation in the 
SC lattice, which is realized in $\EE^3$:
\begin{proposition}\label{prop:SC single}
For a point $z_0\in \EE_\CC$ and
$\delta= (\delta_1, \delta_2, \delta_3)\in \EE^3$ such that
the image of 
the embedding 
$
\iota^\SC_{\delta}:
\cZSC \to 
\cZSC + \delta_\CC$ is a subset of
$\EE_\CC\setminus \{z_0\}$,
$\iota^\SC_{\delta}(\cZSC) \subset \EE_\CC\setminus \{z_0\}$,
the screw dislocation around $z_0$ given by,
$$
\ZZ_{\cZSC, z_0,\delta}^\SC
:= \left( \widehat{\psi}_a^{-1}
(\displace_{z_0, \delta}^\SC(\cZSC))
\right) =
\left( \frac{a}{2\pi\ii} \exp^{-1}
         \left(\displace_{z_0, \delta}^\SC(\cZSC)\right) \right),
$$
is realized in $\EE^3$,
where 
$\gamma_\delta=\exp(2\pi \ii \delta_3/a)$ and 
$\delta_\CC = (\delta_1 + \delta_2 \ii)$.
\end{proposition}

It is worth while noting  that $\ZZ_{\cZSC, z_0,\delta}^\SC$ can be 
regarded as a `covering space'
of the lattice $\cZSC$ and thus there is a natural projection,
$$
\pi_{\cZSC}: \ZZ_{\cZSC, z_0,\delta}^\SC\to \cZSC.
$$
Here each fiber is $\ZZ = \pi_{\cZSC}^{-1}(\ell)$ 
for every $\ell \in \cZSC$.

\subsection{Graph of Screw Dislocation in SC Lattice }\label{section5.3}

We basically consider the local structure of 
$\ZZ_{\cZSC,z0,\delta}^\SC$,
i.e., $\ZZ_{\cZSC,z0,\delta}^\SC\bigcap \pi_{\cZSC}^{-1}U_{\ell a}$ 
for a simply connected neighborhood $U_{\iota_{\delta_\CC}(\ell a)}$ of 
$\iota_{\delta_\CC}(\ell a) \in \EE_\CC\setminus\{z_0\}$ and 
$\ell \in \ZZ[\ii]$.
The $\pi_{\cZSC}^{-1}U_{\iota_{\delta_\CC}(\ell a)}$ can be regarded as 
a {\lq\lq}trivial covering{\rq\rq} as in the sense of Remark \ref{rk:2.1}.
We can continue to consider the edges as in Subsection
\ref{subsection3.2}. The horizontal edges in 
(\ref{eq:graphSC})
can be determined as a set on the same sheet as in Remark \ref{rk:2.1}.
Thus we can consider the graph $G^\SC_{z0, \delta}$ for 
$\ZZ_{\cZSC,z0,\delta}^\SC$ as a natural extension of 
$G^\SC_\delta$.

\subsection{Energy of Screw Dislocation in SC Lattice }\label{section5.4}

Let us consider the graphs $G^\SC_{z0, \delta}$ and $G^\SC_\delta$
as the subsets of $\EE^3$.
Due to the dislocation, the length of each edge in $G^\SC_{z0, \delta}$
is different from that in 
$G^\SC_\delta$. Since $G^\SC_\delta$ is stable mechanically,
the energy of  $G^\SC_{z0, \delta}$ is higher than that of $G^\SC$.
We compute the energy difference following \cite{HMNSU},
which is called the stress energy of screw dislocation or the stress energy
simply.
Further we basically consider the local structure of 
$\ZZ_{\cZSC,z0,\delta}^\SC$ in this section.

In the following, we also assume that $\delta = (0, 0, 0)$ and
 $\gamma_\delta = 1$, and identify $\cZSC$ and its image of 
$\widehat{\iota}^\SC_{\delta}$ for simplicity.
Further we denote $\displace^\SC_{z_0,\delta}$
etc.\  by $\displace^\SC_{z_0}$ etc.\ 
by suppressing $\delta$.

For $\ell \in \ZZ[\ii]$, we define
the relative height differences
$\varepsilon_{\ell}^{(1)}$,
$\varepsilon_{\ell}^{(2)}$ and
$\varepsilon_{\ell}^{(\pm)}$
by 
\begin{equation}\label{eq:dfn_of_eps}
  \begin{array}{rl}
    \displaystyle{\varepsilon_{\ell}^{(1)}} & \displaystyle{= \frac{a}{2\pi\ii}
    \left(\log(\displace^\SC_{z_0}((\ell+1)a)
         -\log(\displace^\SC_{z_0}(\ell a)) \right),} \\
    \displaystyle{\varepsilon_{\ell}^{(2)}} & \displaystyle{= \frac{a}{2\pi\ii}
    \left(\log(\displace^\SC_{z_0}(\ell a+\ii a))
        -\log(\displace^\SC_{z_0}(\ell a)) \right), }\\
     \displaystyle{\varepsilon_{\ell }^{(\pm)}} & \displaystyle{= \frac{a}{2\pi\ii}
    \left(\log(\displace^\SC_{z_0}((\ell+1)a \pm \ii a))
         -\log(\displace^\SC_{z_0}(\ell a) \right), } \\
  \end{array}
\end{equation}
respectively.
It is obvious that for this dislocation of the simple cubic lattice,
$-a/2 < \varepsilon_{\ell}^{(i)} < a/2$
for $i = 1, 2$ and $\pm$.
It is easy to obtain 
\begin{equation}\label{eq:a/z}
  \begin{array}{rl}
    \displaystyle{\varepsilon_{\ell}^{(1)}} & \displaystyle{= \frac{a}{4\pi\ii}
  \left(\log(1 + a/(\ell a-z_0)) - \log(1+\overline{a/(\ell a-z_0)})\right) },
\\
    \displaystyle{\varepsilon_{\ell}^{(2)}} & \displaystyle{= \frac{a}{4\pi\ii}
  \left(\log(1 + a\ii/(\ell a-z_0)) - \log(1+\overline{a\ii/(\ell a-z_0)})\right) },\\
  \displaystyle{\varepsilon_{\ell}^{(\pm)}} & 
\displaystyle{= \frac{a}{4\pi\ii} 
\left(\log(1 + a(1\pm\ii)/(\ell a-z_0))
 - \log(1+\overline{a(1\pm\ii)/(\ell a-z_0)})\right)} .
\end{array}
\end{equation}
Here $\overline{z}$ is the complex conjugate of $z$.
The difference of length $\Delta$
in each segment from the natural
length of $G^\SC_\delta$ is obtained by,
\begin{enumerate}
\item for $[(\ell\pm1,\ell_3)a,(\ell,\ell_3)a]$ and
$[(\ell,\ell_3)a,(\ell+\ii,\ell_3)a]$,
$$
\Delta_{\ell}^{(i)} =\sqrt{a^2 +
(\varepsilon_{\ell}^{(i)})^2}-a, \quad (i=1,2),
$$

\item 
for 
$[(\ell,\ell_3),(\ell+1, \ell_3\pm 1)]$ or 
$[(\ell, \ell_3),(\ell\pm\ii, \ell_3\pm 1)]$
\begin{equation}
\Delta_{\ell}^{d(i, \pm)}
=\sqrt{(a \pm \varepsilon_{\ell}^{(i)})^2
+a^2}-\sqrt{2}a, \quad (i=1,2),
\label{eq:De1}
\end{equation}

\item for
$[(\ell,\ell_3),(\ell+1\pm\ii,\ell_3)]$,
\begin{equation}
\Delta_{\ell}^{d(\pm)} =
\sqrt{2a^2 + (\varepsilon_{\ell}^{(\pm)})^2
}-\sqrt{2}a \quad\mbox{and}
\label{eq:De2}
\end{equation}

\item for
$[(\ell,\ell_3),(\ell,\ell_3+1)]$, $\Delta_{\ell}^{(3)} =0$.

\end{enumerate}

\begin{remark}\label{rmk:epsilon_w}
{\rm{
By letting $\displaystyle{w:=\frac{a}{\ell a - z_0}}$ 
for $\ell \in \ZZ[\ii]$, these
$\varepsilon$'s are real valued functions of $w$ and $\overline{w}$,
i.e,
\begin{gather}
  \begin{split}
 \varepsilon_{\ell}^{(1)}(w, \overline{w}) 
&= \frac{a}{4\pi\ii}\log\left(\frac{1+w}{1+\overline{w}}\right), \quad
 \varepsilon_{\ell}^{(1)}(w, \overline{w}) 
 =\overline{ \varepsilon_{\ell}^{(1)}(w, \overline{w}) },\\
 \varepsilon_{\ell}^{(2)}(w, \overline{w}) 
&= \frac{a}{4\pi\ii}\log\left(\frac{1+\ii w}{1+\overline{\ii w}}\right), \quad
 \varepsilon_{\ell}^{(2)}(w, \overline{w}) 
 =\overline{ \varepsilon_{\ell}^{(2)}(w, \overline{w}) },\\
\varepsilon_{\ell}^{(\pm)}(w, \overline{w}) & 
= \frac{a}{4\pi\ii} 
\log\left(\frac{1+(1\pm\ii) w}{1+\overline{(1\pm\ii)w}}\right), \quad
\varepsilon_{\ell}^{(\pm)}(w, \overline{w})=
\overline{\varepsilon_{\ell}^{(\pm)}(w, \overline{w})}.
\end{split}\label{eq:epsilon_w}
\end{gather}
As these expressions looks simple, 
the origin of the simplicity is the description based on
the elementary number theory which we employ.

Since the square root function $\sqrt{1+x}$ at $x=0$ is also
a real analytic function of $x$, these properties are succeeded to 
these $\Delta$'s.
}}
\end{remark}

\begin{lemma} \label{lm:elasap}
For $\ell a \in \cZSC=\ZZ[\ii]a$ satisfying that 
$\displaystyle{\frac{a}{\sqrt{|\ell a-z_0|^2}}} \ll 1$, 
$\varepsilon^{(1)}_{\ell}$,
 $\varepsilon^{(2)}_{\ell}$
and $\varepsilon^{(\pm)}_{\ell}$
are approximated by
\begin{eqnarray}\label{eq:3.8}
    \varepsilon_{\ell}^{(1)} & = & - \frac{a}{2\pi}
    \frac{a(\ell_2 a-y_0)}{|\ell a-z_0|^2} + o\left(
\frac{a}{\sqrt{|\ell a -z_0|^2}}
\right), \nonumber \\
    \varepsilon_{\ell}^{(2)} & = & -\frac{a}{2\pi}
    \frac{a(\ell_1 a -x_0)}{|\ell a -z_0|^2}  
+ o\left( \frac{a}{\sqrt{|\ell a -z_0|^2}}
\right), \label{eq:approx_of_eps} \\
    \varepsilon_{\ell}^{(\pm)} & = & -\frac{a}{2\pi}
    \frac{(\pm a( \ell_1a-x_0) +a(\ell_2a -y_0))}
          {|\ell a -z_0|^2} 
+ o\left( \frac{a}{\sqrt{|\ell a -z_0|^2}}
    \right), \nonumber
\end{eqnarray}
respectively,
whereas $\Delta_{\ell}^{(i)}$
$\Delta_{\ell}^{d(i, \pm)}$ and
$\Delta_{\ell}^{d(\pm)}$
are approximated by 
\begin{equation}\label{eq:approx_of_delta}
  \begin{array}{l}
    \displaystyle{\Delta_{\ell}^{(i)} = \frac{1}{2a}
    (\varepsilon_{\ell}^{(i)})^2
 + o\left( \left(\frac{a}{\sqrt{|\ell a -z_0|^2}}\right)^2 \right)
 = o\left( \frac{a}{\sqrt{|\ell a -z_0|^2}} \right),}
     \raisebox{0mm}[7mm][7mm]{} \\
     \displaystyle{\Delta_{\ell}^{d(i, \pm)}
= \pm \frac{1}{\sqrt{2}} 
    \varepsilon_{\ell}^{(i)}
+ o\left( \frac{a}{\sqrt{|\ell a -z_0|^2}} \right),}
    \raisebox{0mm}[7mm][7mm]{} \\
    \displaystyle{\Delta_{\ell}^{d(\pm)} = \frac{1}{2\sqrt{2}a} 
    (\varepsilon_{\ell}^{(\pm)})^2
  + o\left( \left(\frac{a}{\sqrt{|\ell a -z_0|^2}}\right)^2 \right)
 = o\left( \frac{a}{\sqrt{|\ell a -z_0|^2}} \right),}
  \end{array}
\end{equation}
respectively, $i = 1, 2$.
\end{lemma}

\begin{proof}
Using $\log (1+z) = z +o(z^2)$, we have the leading terms in (\ref{eq:3.8}).
Remark \ref{rmk:epsilon_w} shows the estimation in (\ref{eq:approx_of_delta}).
However we can also estimate them directly; 
for example, $ (\varepsilon_{\ell}^{(i)})^2$ is estimated by
$$
\left|(\varepsilon_{\ell}^{(i)})^2\right|
=\left|\frac{a^2(\ell_2 a-y_0)^2}{|\ell a-z_0|^4}\right|
\le \left|\frac{a^2}{|\ell a-z_0|^2}\right|,
$$
since $|\ell a-z_0|^2=(\ell_1 a-x_0)^2+(\ell_2 a-y_0)^2$.
Further the relation
$\displaystyle{\sqrt{1+z}-1 = \frac{1}{2}z +o(z^2)}$ shows 
(\ref{eq:approx_of_delta}).
\end{proof}

\begin{lemma}\label{lm:epsilon_c}
We let $\varepsilon_{\ell}^{(c)} :=\varepsilon_{\ell}^{(2)} +
\ii \varepsilon_{\ell}^{(1)}$ and we have the following:
$$
(\varepsilon_{\ell}^{(1)})^2+
(\varepsilon_{\ell}^{(2)})^2 
= \varepsilon_{\ell}^{(c)}\overline{\varepsilon_{\ell}^{(c)}}
=\frac{a^2}{4\pi^2}
    \frac{a^2}{|\ell a -z_0|^2}  
+ o\left( \frac{a}{\sqrt{|\ell a -z_0|^2}}\right).
$$
\end{lemma}

\begin{remark}\label{rmk:epsilon_c}
{\rm{
As the term in Lemma \ref{lm:epsilon_c} is the leading term in
the stress energy in Theorem \ref{thm:energySC}, 
it represents the symmetry of the
dislocation in the SC-lattice as follows.
\begin{enumerate}
\item  Since
the lattice points are expressed in terms of the Gauss integers
$\ZZ[\ii]$, the edges in the graph $G^\SC_{z_0, \delta}$ are described by 
$(\ell, \ell + d)$ for $d = \{\pm 1, \pm \ii\}$, and 
\item the real valued analytic function in $f(w, \overline{w})$
in the complex structure in plane $w \in \CC$ has the property,
$f(w, \overline{w}) = \overline{f(w, \overline{w})}$.
\end{enumerate}
Our expression in terms of $\ZZ[\ii]$
 shows the property manifestly.

Further in the computations, we use the relation
$1+\ii^2=0$, which comes from the cyclotomic symmetry
as mentioned in Remark \ref{rmk:AlgSC}.
}}
\end{remark}

Following \cite{HMNSU}, 
let us introduce the subsets of $\ZZ[\ii]$,
\begin{equation}
  A_{\rho, N}^{\ii} := \left\{ \ell \in \ZZ[\ii]\,\Bigr|\,
\rho a < |\ell a -z_0|
< N a \right\} \subset \ZZ[\ii]
\label{eq:ArhoN}
\end{equation}for $N > \rho$, which is bounded
and is a finite set, and the core region $C^\ii_\rho$,
$$
C^\ii_\rho:= \left\{ \ell \in \ZZ[\ii]\,\Bigr|\,
|\ell a -z_0|\le \rho a  \right\} \subset \ZZ[\ii].
$$
Let $A_{\rho}^{\ii} :=\displaystyle{\lim_{N\to \infty}A_{\rho, N}^{\ii}}$.

Let us evaluate the stress energy, the elastic energy 
caused
by the screw dislocation
 in the meso-scopic scale.
Since the screw dislocation  is invariant under
the translation
from  $\ell_3$
to $\ell_3 +1$,
we compute the energy density for unit length in the
$(0,0,1)$-direction using Remark \ref{rk:2.2} and call it simply the stress
energy of dislocation again.

Let $k_p$ and $k_d$ be the spring constants of the horizontal
springs and the diagonal springs respectively.
Then, the stress
energy of dislocation in the 
annulus region
$A_{\rho,N}^\ii$ is given by
\begin{equation}
E^\SC_{\rho,N}(z_0) := \sum_{\ell \in A_{\rho,N}^\ii}
\cE^\SC_{\ell},
\label{eq15aa}
\end{equation}
where 
$\cE^\SC_{\ell}$ is the energy density
defined by 
\begin{eqnarray}
\cE^\SC_{\ell}& := &
\frac{1}{2}k_p
\biggl(\left(\Delta_{\ell}^{(1)}\right)^2
+\left(\Delta_{\ell}^{(2)}\right)^2 \biggr)
 + \frac{1}{2}k_d
\biggl(
\left(\Delta_{\ell}^{d(1, +)}\right)^2
+\left(\Delta_{\ell}^{d(2, +)}\right)^2
+\left(\Delta_{\ell}^{d(1, -)}\right)^2  \nonumber \\
& & \qquad\qquad\qquad
+\left(\Delta_{\ell}^{d(2, -)}\right)^2
+\left(\Delta_{\ell}^{d(+)}\right)^2
+\left(\Delta_{\ell}^{d(-)}\right)^2\biggr).
\raisebox{0mm}[4mm][4mm]{}. 
\nonumber
\end{eqnarray}

We recall Proposition 9 in \cite{HMNSU} in terms of our convention,
which are also directly obtained via Lemmas \ref{lm:elasap} and
\ref{lm:epsilon_c}, and 
Remarks \ref{rmk:epsilon_w} and \ref{rmk:epsilon_c}:

\begin{proposition} \label{prop:9}
\begin{itemize}
\item[$(1)$] For $\ell \in A_\rho^\ii$,
the energy density $\cE^\SC_{\ell}$ is
expressed by a real analytic function $\cE^\SC(w, \overline{w})$ of
$w$ and $\bar{w} \in \CC$ with $|w| < 1/\sqrt{2}$ in such a way that
\begin{equation*}
\cE^\SC_{\ell} = \cE^\SC\left(
\frac{a}{\ell a - z_0},
\frac{a}{\overline{\ell a - z_0}}
\right).
\end{equation*}
\item[$(2)$] For the power series expansion
$$
\cE^\SC(w,\overline w) = \sum_{s=0}^\infty
\cE_\SC^{(s)}(w, \overline w), \quad
\cE_\SC^{(s)}(w, \overline w) := \sum_{i+j=s, i,j\ge0} C_{i, j} 
w^i \overline{w}^j,
$$
with  $C_{i, j} \in \CC$, the following holds
\begin{itemize}
\item[\textup{(a)}]  
$\cE_\SC^{(0)}(w, \overline w)=\cE_\SC^{(1)}(w, \overline w)=0$,
\item[\textup{(b)}]  the leading term is given by 
\begin{equation}
\cE_\SC^{(2)}(w, \overline w) =
\frac{a^2}{8\pi^2} k_d 
w\overline{w},
\qquad
\cE_\SC^{(2)}\left(
\frac{a}{\ell a - z_0},
\frac{a}{\overline{\ell a - z_0}}
\right)=
\frac{1}{8\pi^2} k_d
 \left[
  \frac{a^4}{|\ell a - z_0|^2 }
 \right], 
\label{eq15a}
\end{equation}
\item[\textup{(c)}]  $C_{i, j}=\overline{C_{j, i}}$, and
\item[\textup{(d)}]  for every $s \geq 2$, there is a constant $M_s > 0$ such that
$$
|\cE_\SC^{(s)}(w, \overline w)| \le M_s |w|^s. 
$$
\end{itemize}
\end{itemize}
\end{proposition}

\bigskip

As the summation in (\ref{eq15aa}) is finite, we have
\begin{equation}
E^\SC_{\rho,N}(z_0) =   
\sum_{s=2}^\infty
\sum_{\ell\in A_{\rho,N}^\ii}
\cE_\SC^{(s)}\left(
\frac{a}{\ell a - z_0},
\frac{a}{\overline{\ell a - z_0}}
\right).
\label{eq-series}
\end{equation}

As mentioned in Lemmas \ref{lm:elasap},
following \cite{HMNSU},
the 
``principal part'' of the stress energy of 
the screw dislocation in the SC lattice is given by
the following theorem:
\begin{theorem}\label{thm:energySC}
The principal part of the stress energy $E_{\rho,N}(z_0)$, defined by
\begin{equation*}
  E^{\SC(\mathrm{p})}_{\rho,N}(z_0) :=
   \sum_{\ell \in A_{\rho,N}^\ii} 
\cE_\SC^{(2)}\left(
\frac{a}{\ell a - z_0},\frac{a}{\overline{\ell a - z_0}}
\right)  = \frac{1}{8\pi^2} k_d
   \sum_{\ell \in A_{\rho,N}^\ii}
   \left[
    \frac{a^4}{|\ell a -z_0|^2}
   \right]
\end{equation*}
is given by the truncated
Epstein-Hurwitz zeta function (see Appendix),
\begin{equation}
E^{\SC(\mathrm{p})}_{\rho,N}(z_0)
= \frac{1}{8\pi^2}k_d a^2
\zeta_{A_{\rho, N}^\ii}^{\ii}(2, -z_0/a).
\label{eq:th1}
\end{equation}
\end{theorem}

As mentioned in Remark \ref{rmk:AlgSC}, it is noted that
we obtain
(\ref{eq15a}) and this theorem due to the cyclotomic symmetry
(\ref{eq:AlgSC}).

By Proposition~\ref{prop:9} (2) (d),
we can estimate each of the other terms appearing
in the power series expansion (\ref{eq-series})
by the truncated
Epstein-Hurwitz zeta function  as follows.

\begin{proposition}\label{prop:remainder}
For each $s \geq 3$, there exists a positive constant $M_s'$ such that 
\begin{equation}
\sum_{\ell\in A_{\rho,N}^\ii}
\cE_\SC^{(s)}\left(
\frac{a}{\ell a - z_0}, \frac{a}{\overline{\ell a - z_0}}
\right)  \le M_s' 
\zeta_{A_{\rho, N}^\ii}^{\ii}(s, -z_0/a).
\end{equation}
\end{proposition}

\section{Screw Dislocation in BCC Lattice and its energy} \label{section4}

In this section, we consider the screw dislocation in the BCC lattice.
The studies of the screw dislocations in the
BCC crystal lattices have long history, e.g.,
\cite{KV,Tak}, and still attract attentions;
some of the studies are 
based on the first principle approaches, e.g. \cite{Cl,GD,IKY},
others are via the continuum approaches, e.g, \cite{Mer, KE, KF, LA},
and geometrical approaches  \cite{P, AO, ATMO, HO}.

However in this paper, we concentrate ourselves on the number theoretic
approaches based on the previous report \cite{HMNSU}.
We summarize the algebraic descriptions of the BCC lattice 
and its screw dislocation in \cite{HMNSU}.
In this paper, we employ the novel description of 
the screw dislocation in terms of the elementary number theory.
We show that the screw dislocation of the BCC lattice is expressed well
 in terms of the Eisenstein integers.
Using this description, we compute its stress or the stress energy 
like the case of the SC lattice.

\subsection{Preliminary: Eisenstein Integers}

We show the basic properties of the 
Eisenstein integers $\ZZ[\omega_6]=\ZZ[\omega_3]$ \cite{Tri} (see Appendix).
For the primitive sixth root of unit, $\omega_6$, we have the following 
relations:
\begin{lemma}\label{lm:cyl}
$$
1+\omega_6^2+\omega_6^4=0, \quad -\omega_6=\omega_6^4, \quad 
\overline{\omega_6}=\omega_6^5.
$$
\end{lemma}

\bigskip

We introduce  $\nu_i$ and $\mu_i$ by
\begin{equation}
\nu_i := \frac{1}{3}(\omega_6^i + \omega_6^{i+1}), \quad i = 0, 1, 2,\ldots, 5,
\quad \mu_0 := 0, \quad \mu_1 :=\nu_0, \quad \mu_2:=\nu_1.
\label{eq:mus}
\end{equation}
It is noted that they belong to 
$\dfrac{1}{3}\ZZ[\omega_6]:=\{\ell_1 + \ell_2 \omega_6\ | \ 
3 \ell_a \in \ZZ\}$ and have the properties in the following lemma:

\begin{lemma}
\label{lm:4.2}
\begin{enumerate}
\item $\ZZ[\omega_6]= \ZZ \oplus \ZZ \omega_6=\ZZ[\omega_3]$,

\item $\ZZ[\omega_6]+\nu_0\ni \nu_2,\nu_4$, 

\item $\ZZ[\omega_6]+\nu_1\ni \nu_3, \nu_5$, and

\item for $z \in \EE_\CC$, 
$$
\sum_{i=0}^2 \left(\frac{\nu_{2 i}}{z}
-\frac{\overline{\nu_{2 i}}}{\overline{z}}\right)^2
=-2\frac{1}{|z|^2},
$$
$$
\sum_{i=0}^2 \left(\frac{\nu_{2 i+1}}{z}
-\frac{\overline{\nu_{2 i+1}}}{\overline{z}}\right)^2
=-2\frac{1}{|z|^2}.
$$
\end{enumerate}
\end{lemma}

\begin{proof}
The relations (1)-(3) are geometrically obvious but
it can, also, be proved by the cyclotomic properties in Lemma 
\ref{lm:cyl}.
In (4), the left hand side is equal to
$$
\sum_{i=0}^2 \frac{(\nu_{2 i} \overline{z}-\overline{\nu_{2 i}} z)^2}
{(|z|^2)^2} =
\frac{1}{(|z|^2)^2}
\left(\overline{z}^2 \sum_{i=0}^2\overline{\nu_{2 i}}^2-2 |z|^2
\sum_{i=0}^2|\nu_{2 i}|^2 +
{z}^2 \sum_{i=0}^2\nu_{2 i}^2\right).
$$
From Lemma \ref{lm:cyl}, we have
$$
\sum_{i=0}^2\overline{\nu_{2 i}}^2=0, \quad
\sum_{i=0}^2\nu_{2 i}^2=0
$$
and thus the left hand side gives the right hand side.
\end{proof}

\begin{remark}\label{rmk:fC6action0}
{\rm{
As mentioned in Remark \ref{rmk:AlgSC}, $\ZZ[\ii]$ 
has the cyclic group $\fC_4$ action of the order 4 
as the cyclotomic symmetry of $\ZZ[\ii]$.
The Eisenstein integers $\ZZ[\omega_6]$ has
the cyclic group $\fC_6$ action of the order 6
as the cyclotomic symmetry of $\ZZ[\omega_6]$.
Since the cyclic group $\fC_3$ of the order 3 is 
a subgroup of $\fC_6$, there is the 
$\fC_3$ action on $\ZZ[\omega_6]$.
Lemma \ref{lm:4.2} (4) is based on the symmetry.

As the  cyclotomic symmetry of $\ZZ[\ii]$ plays the
important role
in Lemma \ref{lm:epsilon_c} and Remark \ref{rmk:epsilon_c},
and thus in Proposition \ref{prop:9} and Theorem \ref{thm:energySC},
the cyclotomic symmetry in $\ZZ[\omega_6]$ also plays an
important role in the evaluation of the stress energy
of the dislocation in BCC lattice
as in Lemma \ref{lm:elasap2},
Remark \ref{rmk:fC6action1}, 
Proposition \ref{prop:9B},
and Theorem \ref{thm:energyBCC}.
}}
\end{remark}

\subsection{Algebraic Structure of BCC Lattice}

Though there are several algebraic descriptions of the BCC lattice 
(see \cite[p.~116]{CS}, for example),
we recall the algebraic descriptions of the BCC lattice \cite{HMNSU}.
We assume  that $a_1=(a,0,0)$, $a_2=(0,a,0)$, $a_3=(0,0,a)$
in  $\RR^3$ for a positive real number $a$ as shown in
 Figure~\ref{fig:BCC01}.
The generator $b$ corresponds to the center point of the cube generated by
$a_1$, $a_2$ and $a_3$.
The BCC lattice is the lattice
in $\RR^3$ generated by $a_1$, $a_2$, $a_3$ and
$b=(a_1+a_2+a_3)/2$.
Algebraically, it
is described as an additive group (or a $\ZZ$-module) by
$$
\BB^a := \LA a_1, a_2, a_3, b\RA_\ZZ/\LA 2b-a_1-a_2-a_3 \RA_\ZZ,
$$
where $\LA 2b-a_1-a_2-a_3 \RA_\ZZ$ is the subgroup generated by
$2b-a_1-a_2-a_3$. The lattice point in $\BB^a$ is given by
$\ell_1 a_1 + \ell_2 a_2 + \ell_3 a_3 + \ell_b b $
for a certain $\ell_i \in \ZZ$ ($i=1, 2, 3$) and 
$\ell_b\in \{0, 1\}$.

\begin{figure}[t]
 \begin{center}
\includegraphics[width=8cm]{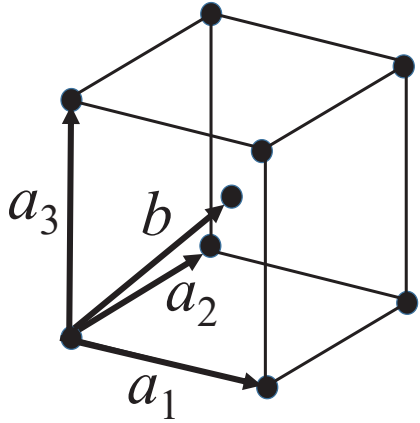}
 \end{center}
\caption{
BCC lattice: The unit cell of the BCC lattice
is illustrated by $a_1$, $a_2$, $a_3$ and $b$, where
$b = (a_1 + a_2 + a_3)/2$.}
\label{fig:BCC01}
\end{figure}

\begin{figure}[t]
 \begin{center}
\includegraphics[width=15cm]{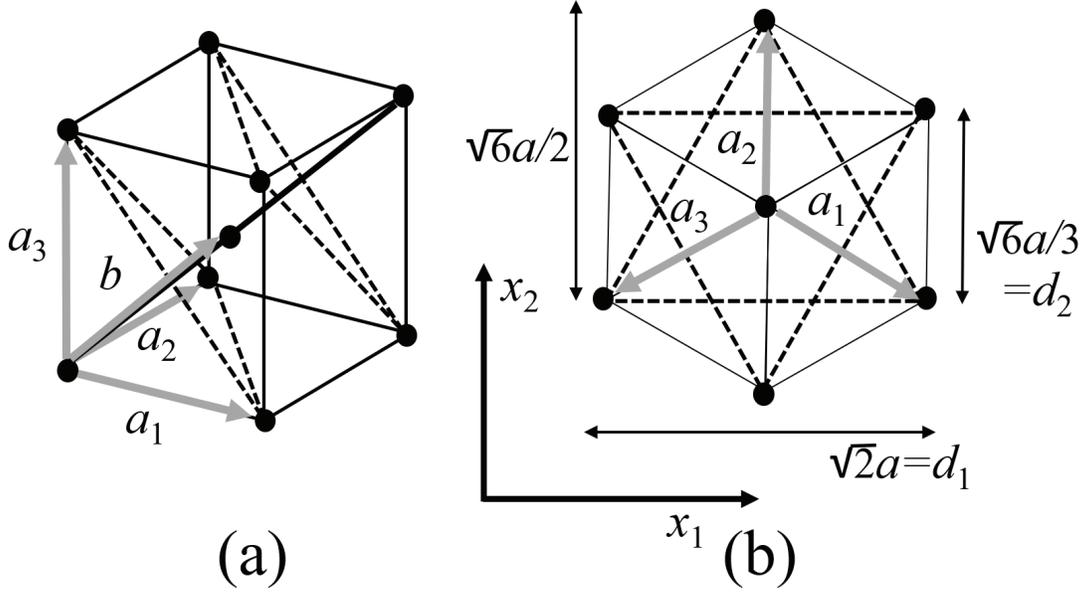}
 \end{center}
\caption{
BCC lattice and its projection along the 
$(1,1,1)$-direction: (a) shows the panoramic view of the 
 unit cell of the BCC lattice which contains two triangles whose 
normal direction is $(1,1,1)$. (b) shows its projection along the 
$(1,1,1)$-direction corresponding to (a).}
\label{fig:BCC02}
\end{figure}

\begin{figure}[t]
 \begin{center}
\includegraphics[width=10cm]{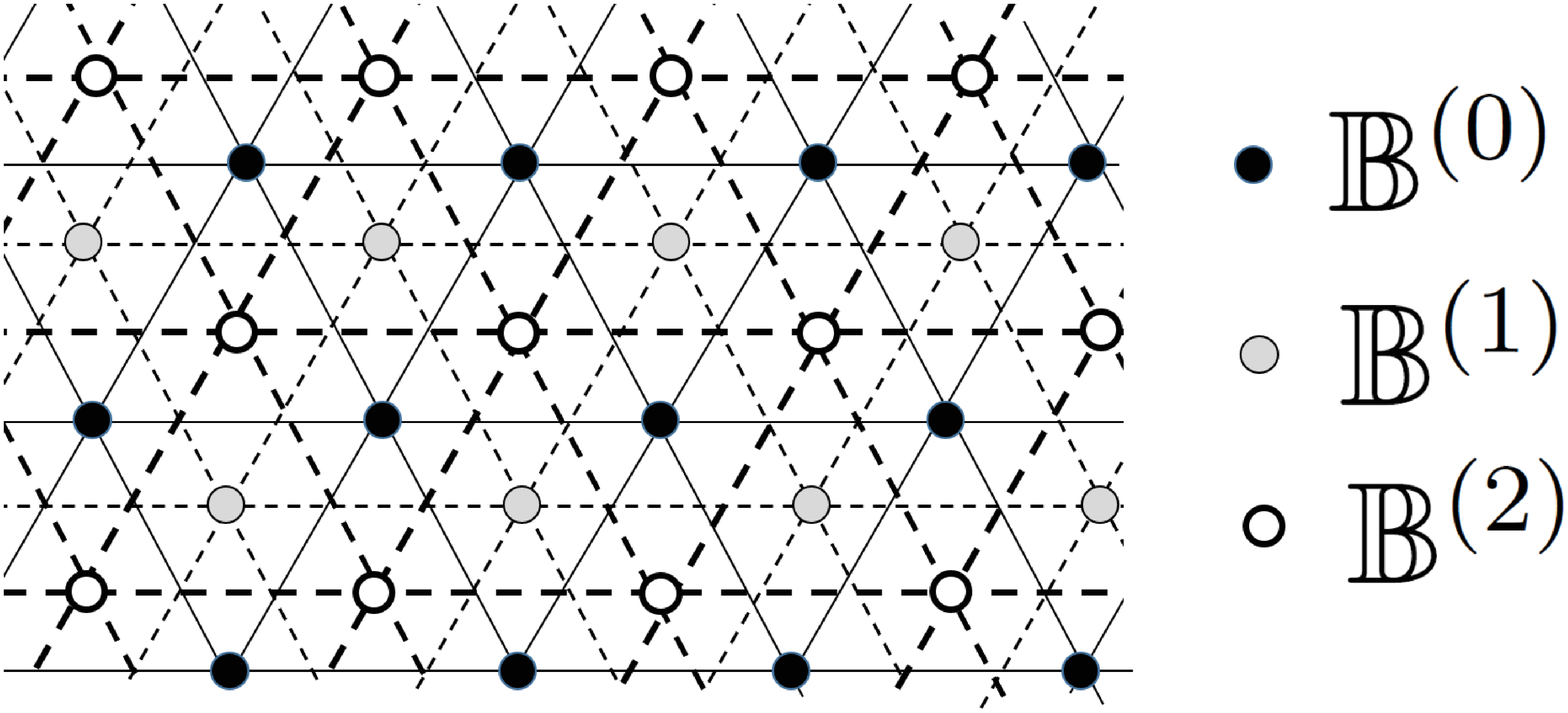}
 \end{center}
\caption{
BCC lattice: The black, gray and white dots
correspond to the three sheets $\BB^{(0)}$, $\BB^{(1)}$ and $\BB^{(2)}$,
which are associated with  
$\cZBCC{(0)}$, $\cZBCC{(1)}$ and 
$\cZBCC{(2)}$ respectively. }
\label{fig:BCC03}
\end{figure}

The lattice $\BB^a$ is group-isomorphic to the multiplicative group
$$
\BB
:= \{\alpha_1^{\ell_1}\alpha_2^{\ell_2}\alpha_3^{\ell_3}\beta^{\ell_4} \,
| \, \mbox{abelian}, \ell_1, \ell_2, \ell_3, \ell_4 \in \ZZ, \,
\beta^2 \alpha_1^{-1} \alpha_2^{-1} \alpha_3^{-1} =1\}.
$$
Let us denote by $\AA_4$ the multiplicative free abelian group
of rank $4$ generated by
$\alpha_1$, $\alpha_2$, $\alpha_3$ and $\beta$, i.e.,
$$
\AA_4 := \{\alpha_1^{\ell_1}\alpha_2^{\ell_2}\alpha_3^{\ell_3}\beta^{\ell_4} \,
| \, \mbox{abelian}, \, \ell_1, \ell_2, \ell_3, \ell_4 \in \ZZ\}.
$$
Then,
$\BB$ is also described as the quotient group
$$
\BB = \AA_4/\LA \beta^2 \alpha_1^{-1} \alpha_2^{-1} \alpha_3^{-1} \RA,
$$
where $\LA \beta^2 \alpha_1^{-1} \alpha_2^{-1} \alpha_3^{-1} \RA$
is the (normal) subgroup generated by
$\beta^2 \alpha_1^{-1} \alpha_2^{-1} \alpha_3^{-1}$.
We shall consider the group ring
$\CC[\BB]$ of $\BB$,
$$
\cR_6:=
\CC[\BB] = \CC[
\alpha_1,\alpha_2, \alpha_3,
\alpha_1^{-1},\alpha_2^{-1}, \alpha_3^{-1},
\beta, \beta^{-1}]/
(\beta^2 - \alpha_1 \alpha_2 \alpha_3).
$$

\subsection{Algebraic Structure of BCC Lattice for $(1,1,1)$-Direction}

It is known  that a screw dislocation in the BCC lattice is basically given by
the $(1,1,1)$-direction since the Burgers vector is parallel to the
$(1,1,1)$-direction \cite{N}.

In this subsection, we consider the algebraic
structure of the BCC lattice
of (111)-direction  to describe its fibering structure
by noting Figure \ref{fig:BCC02} (a) and (b).
Let us consider the subgroup of $\BB$,
which corresponds to the translation in the plane vertical to 
$(1,1,1)$-direction,
$$
\BB_H:=\{
(\alpha_1\alpha_3^{-1})^{\ell_1}
(\alpha_2\alpha_3^{-1})^{\ell_2} \, | \,
\ell_1, \ell_2 \in \ZZ\},
$$
and  $\CC[\BB_H]$-modules.

\begin{lemma}\label{lm:BCC}
There are  isomorphisms as $\CC[\BB_H]$-modules:
\begin{eqnarray*}
\cR_6/(\alpha_1\alpha_2\alpha_3-1)
& \cong &
\CC[\BB_H]
\oplus\CC[\BB_H]\alpha_1
\oplus\CC[\BB_H]\alpha_1\alpha_2\\
& & \quad \oplus
\CC[\BB_H]\beta
\oplus\CC[\BB_H]\alpha_1\beta
\oplus\CC[\BB_H]\alpha_1\alpha_2\beta.
\end{eqnarray*}
$$
\cR_3:=\cR_6/(\beta-1) \cong
\CC[\BB_H]
\oplus\CC[\BB_H]\alpha_1
\oplus\CC[\BB_H]\alpha_1\alpha_2.
$$
\end{lemma}

These decompositions mean that the BCC lattice has 
the triple different fiber structures of three sheets.
We should note that $\cR_6$ can be regarded as a double covering
of $\cR_3$.
The interval between the sheets is now given by $\sqrt{3}a/6$,
and let us denote $\cR_3$ as a set,
the image of the forgetful functor
to the category of set, by $\BB^a$ as a
 subset of the vector space $\RR^3$
corresponding to
the three sheets:
\begin{lemma}
As a set, $\BB^a$ is also decomposed as
$$
\BB^a = \BB^{(0)} \coprod \BB^{(1)} \coprod \BB^{(2)} ,
$$
where
$$
\begin{array}{rl}
      \BB^{(0)}&:=\{ \ell_1(a_1-a_3) +\ell_2(a_2-a_3) +\ell_3 b \,
                 | \, \ell_1, \ell_2, \ell_3 \in \ZZ\}\subset \RR^3, \\
      \BB^{(1)}&:=\{ \ell_1(a_1-a_3) +\ell_2(a_2-a_3) + a_1 +\ell_3 b \,
                 | \, \ell_1, \ell_2, \ell_3 \in \ZZ\}\subset \RR^3,\\
      \BB^{(2)}&:=\{ \ell_1(a_1-a_3) +\ell_2(a_2-a_3)+ a_1+ a_2+\ell_3 b \,
                 | \, \ell_1, \ell_2, \ell_3 \in \ZZ\}\subset \RR^3.\\
\end{array}
$$
\end{lemma}

\subsection{Fiber Structure of BCC Lattice 
and Eisenstein Integers}

We can regard $\BB^{(a)}$ as trivial covering space of $\ell_3$-direction.
On the other hand, the additive group of $\BB_H$,
$$
\BB_H^a:=\{ \ell_1(a_1-a_3) +\ell_2(a_2-a_3)  \,
                 | \, \ell_1, \ell_2\in \ZZ\}\subset \RR^2,
$$
can be expressed by the the Eisenstein integers.

We define
$$
d_0:=\frac{\sqrt{3}}{2}a=|b|,\quad
d_1 := \sqrt{2} a, \quad 
d_2:=\frac{1}{\sqrt{3}}d_1=\frac{\sqrt{2}}{\sqrt{3}} a,\quad 
d_3:=\frac{\sqrt{3}}{6}a=\frac{d_0}{3}.
$$
and $\cZBCC{(a)} := (\ZZ[\omega_6]+\mu_c)d_1$ $(a=0,1,2)$ using (\ref{eq:mus}),
 i.e.,
$$
   \cZBCC{(0)} = \ZZ[\omega_6]d_1, \quad
   \cZBCC{(1)} = \ZZ[\omega_6]d_1+\nu_0d_1, \quad  
   \cZBCC{(2)} = \ZZ[\omega_6]d_1+\nu_1d_1, \quad  
$$
which correspond to
$\BB^{(0)}$, $\BB^{(1)}$ and $\BB^{(2)}$ respectively
as in Figure \ref{fig:BCC03}, i.e., there are 
natural projections for $b$-direction,
 $\pi_{\BCC}^{(a)}: \BB^{(a)}\to \cZBCC{(a)}$.

\begin{remark}\label{rmk:AlgBCC}
{\rm{
The projection $\pi_{\BCC}^{(a)}$ of 
these $\BB^{(0)}$, $\BB^{(1)}$ and $\BB^{(2)}$ are essentially equal to 
$\ZZ[\omega_6]$ up to translation and dilatation $d_1$.
Further $\mu_c$ and $\nu_i$ are third points in the lattice
$\nu_i \in \frac{1}{3}\ZZ[\omega_6]$.
They have the algebraic properties of 
Lemmas \ref{lm:cyl} and \ref{lm:4.2} and Remark \ref{rmk:fC6action0}, 
whose origin in  
the ring of integers of the cyclotomic field $\QQ[\omega_6]$.
They have studied in the number theory and algebraic geometry
\cite[Appendix]{FKMPA},
and their application to physics \cite{Ma}. 

As shown in the following,
the $z_3$ position of each sheet in the screw dislocations along the 
$(1,1,1)$-direction and 
the local energy due to the dislocation can be regarded as functions on
$\ZZ[\omega_6]$, more precisely on $\cZBCC{(c)} := (\ZZ[\omega_6]+\mu_c)d_1$ 
$(c=0,1,2)$.
Thus Lemmas \ref{lm:cyl} and \ref{lm:4.2} govern the 
computations of the functions in Lemma \ref{lm:elasap2}
and make them very simple and connected with the 
Epstein-Hurwitz zeta function as in 
Theorem \ref{thm:energyBCC}.

Though in the previous works including \cite{HMNSU}, 
such properties had not been  mentioned,
these properties show algebraic nature of the BCC lattice and 
the screw dislocation in the BCC lattice.
}}
\end{remark}

For a point $\delta = (\delta_1, \delta_2, \delta_3) \in \EE^3$,
we consider the embedding $\iota_\delta :\BB^a \to \EE^3$
and its image $\iota_\delta(\BB^a)$.
Corresponding to $\iota_\delta$,
for the point $\delta_\CC = \delta_1+\ii \delta_2\in \EE_\CC$,
let us also consider the embedding $\iota_{\delta_\CC}:
\cZBCC{(c)} \to \cZBCC{(c)}+\delta_\CC \in \EE_\CC$
and the bundle maps $\widehat\iota_{\delta_\CC}$.
Further 
for $\gamma_{\delta}
=\ee^{\ii \delta_3 /d_0}\in S^1$ and a constant 
section $\displace_{\delta}\in \Gamma(\EE_\CC, S^1_{\EE_\CC})$
$(\displace_{\delta}(z)=\gamma_{\delta})$, 
we consider $\displace^\BCC_\delta \in \Gamma(\cZBCC{}, S^1_{\cZBCC{}})$
as
$$
\displace^\BCC_\delta:=\displace_{\delta} \circ\iota_{\delta_\CC}.
$$
where 
$S^1_{\cZBCC{}}$ is the trivial $S^1$ bundle over $\cZBCC{}$.

\begin{proposition}\label{prop:BCClattice}
The BCC lattice $\iota_\delta(\BB^a)$ is expressed by 
$$
\bigcup_{c=0}^2
\widehat\iota_{\delta_\CC} \left(
\widehat{\psi}_{d_0}^{-1}
\left(\omega_3^{-c} \displace^\BCC_\delta(\cZBCC{(c)})\right)\right)
\subset \EE_{\EE_\CC}=\EE^3.
$$
\end{proposition}
Here we set $d=d_0=|b|$ of $\psi_d$ in Section 2.

\subsection{Spiral Structure in Graph of BCC Lattice}
\label{subsec:Spiral}

\begin{figure}[t]
\begin{center}
\includegraphics[width=13cm]{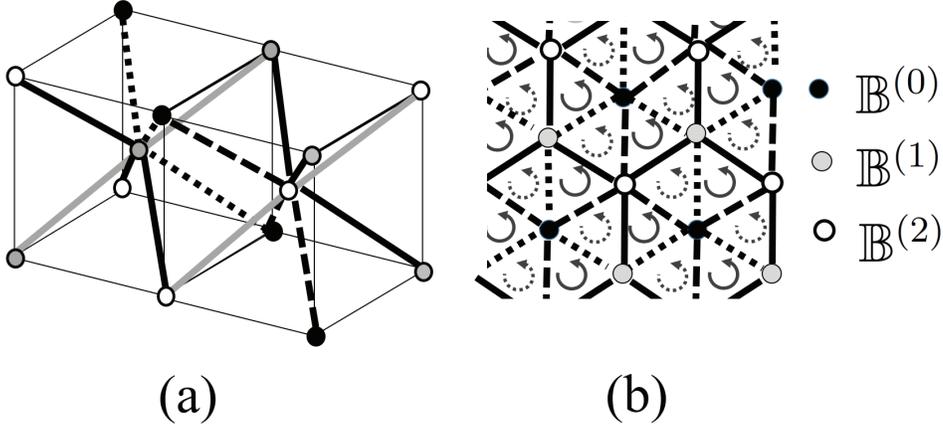}
\end{center}
\caption{
Spiral Structure in Graph of BCC Lattice:
(a) is a panoramic view of a part of $G^\BCC_\delta$ consisting of
two unit cells of the BCC lattice with edges,
where the gray line shows the $(1,1,1)$-direction,
and the black lines means the edges of $G^\BCC_\delta$.
The edges drawn by the short dotted line, long dotted line
and solid line corresponds to those of the projected graph
$\pi_G G^\BCC_\delta$ as in (b).
The arrows of solid lines show the ascendant triangles whereas
the arrows of dot lines correspond to the 
descendant triangles
}
\label{fig:BCC04}
\end{figure}

In the BCC lattice $\iota_\delta(\BB^a)$, 
let us consider the graph $G^\BCC_\delta$ whose
nodes are given as the lattice points of the BCC lattice
and 
edges are given as the shortest connections of the nodes
as shown in Figure \ref{fig:BCC04} (a).
We regard $G^\BCC_\delta$ as a subset of $\EE^3$.
The 0-th sheet $\iota_\delta(\BB^{(0)})$ 
($\iota_{\delta_\CC}(\cZBCC{(0)})$) whose nodes 
are denoted by the black dots is connected with 
the first sheet $\iota_\delta(\BB^{(1)})$ ($\iota_{\delta_\CC}\cZBCC{(1)})$) 
which corresponds to the gray dots, via the long dots lines
in Figure \ref{fig:BCC04} (a). The short dot line
connects 
$\iota_\delta(\BB^{(2)})$ ($\iota_{\delta_\CC}\cZBCC{(2)})$)
and $\iota_\delta(\BB^{(0)})$ ($\iota_{\delta_\CC}\cZBCC{(0)})$)
 whereas
the black lines
connect 
$\iota_\delta(\BB^{(1)})$ ($\iota_{\delta_\CC}\cZBCC{(1)})$) and
$\iota_\delta(\BB^{(2)})$ ($\iota_{\delta_\CC}\cZBCC{(2)})$).

The graph $G^\BCC_\delta$ has a projection to a plane $\EE_\CC$ as in 
Figure \ref{fig:BCC04} (b): $\pi_G: G^\BCC_\delta\to \EE_\CC$,
These edges give the paths which connect these covering sheets.
As in Figure \ref{fig:BCC04} (b), let us consider the path whose projection
is a cycle $\pi_G(G^\BCC_\delta)$ consisting of tree edges, which is called 
"spiral path" because the end point $p$ and the start point $q$ exist
on the different covering sheets but $q \in \pi_G^{-1}(\pi_G(p))$;
if the start point is $(\ell_1, \ell_2, \ell_3)$ in $\iota_\delta(\BB^{(a)})$,
the end point is given as $(\ell_1, \ell_2, \ell_3')$ in 
$\iota_\delta(\BB^{(a)})$
for $|\ell_3' - \ell_3|=d_0=|b|$.
Thus the path shows the spiral curve in $\EE^3$.

The set of the spiral paths is classified by two types.
We assign an orientation on $\EE_\CC$ and the orientation of the 
arrowed graph $[\pi_G(G^\BCC_\delta)]$ is naturally induced from it. 
For the oriental cycle in $\pi_G(G^\BCC_\delta)$, the spiral path is 
ascendant or decedent with respect to $\ell_3$.
We call these triangle cells {\it{ascendant cell}} and
{\it{descendant cell}} respectively. 
They are illustrated in Figure \ref{fig:BCC05} (a) and (b) respectively.

For a center point $z_c$ of a ascendant triangle cell of 
$\pi_G(G^\BCC_\delta)$, 
the nodes in $G^\BCC_\delta$ are given by
\begin{equation}
        \psi_{d_0}^{-1}\left(\gamma_\delta\frac{z-z_c}{|z-z_c|}\right),
\label{eq:ascendent}
\end{equation}
whereas for a center point $z_c$ of a descendant triangle cell of 
$\pi_G(G^\BCC_\delta)$, 
the nodes are expressed by
\begin{equation}
        \psi_{d_0}^{-1}\left(\gamma_\delta\frac{\overline{z-z_c}}{|z-z_c|}\right).
\label{eq:descendent}
\end{equation}

These pictures are well-described in the works of Ramasubramaniam, 
Ariza and Ortiz \cite{ARO} and \cite{AO} using the homological 
investigations more precisely.

\begin{figure}[t]
\begin{center}
\includegraphics[width=13cm]{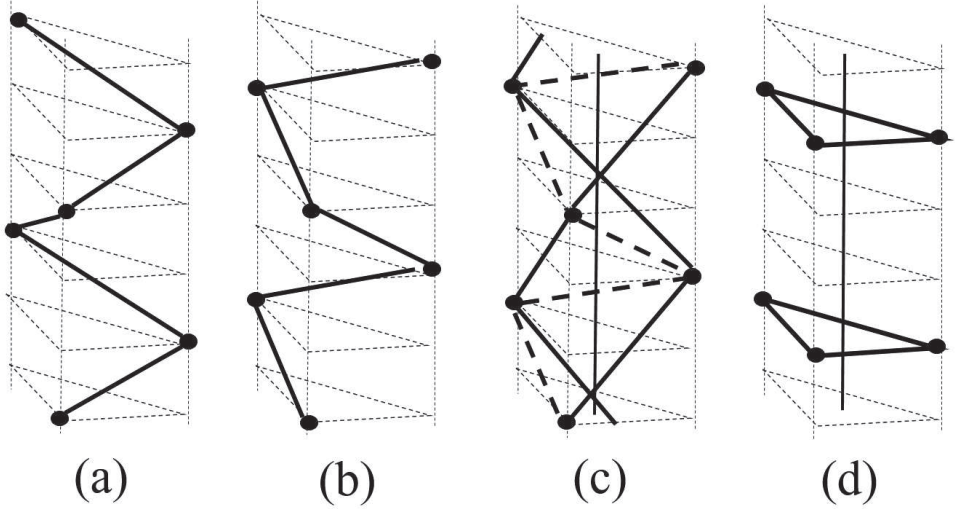}
\end{center}
\caption{
Spiral Paths in BCC lattice: 
(a) and (c) are ascendant spiral paths
and  (b) and (d) decedent spiral paths.
(a) and (b) are normal cases whereas
(c) and (d) are the behavior when in the center, the 
screw dislocation exists.}
\label{fig:BCC05}
\end{figure}

\subsection{Algebraic Description of Screw Dislocations in BCC Lattice}\label{subsection4.4}

As defined in Subsection \ref{subsec:Spiral}, 
for $\delta \in \EE^3$ and $z_0 \in \EE_\CC$,
we use the embedding $\iota_{\delta_\CC}:
\cZBCC{(i)} \to \cZBCC{(i)}+\delta_\CC \in \EE_\CC\setminus\{z_0\}$,
$\delta_\CC = \delta_1+\ii \delta_2\in \EE_\CC$,
 $\gamma_{\delta}=\ee^{\ii \delta_3 /d_0}\in S^1$
and the bundle map $\widehat{\psi}_{d_0}$ so that 
the description of the screw dislocation is obtained as follows.
Let us consider the non-trivial $S^1$-bundle over $\cZBCC{}$
induced from the embedding $\iota_{\delta_\CC}$.
Using the section of $\displace_{z_0, \delta} \in
\Gamma(\EE_\CC\setminus \{z_0\}, S^1_{\EE_\CC\setminus \{z_0\}})$ in
(\ref{eq:sigmaz0}),
we define the section $\widehat\displace_{z_0,\delta}$
in $\Gamma(\cZBCC{}, S^1_{\cZBCC{}})$,
$$
\displace^\BCC_{z_0,\delta}=\iota^{*}_{\delta_\CC}\displace_{\delta,z_0}
=
\displace_{z_0,\delta}\circ \iota_{\delta_\CC}.
$$
It implies that 
$$
\displace^\BCC_{z_0,\delta}(\ell d_1) 
=  \gamma_\delta\frac{ d_1 \ell+\delta_\CC -z_0}
        {|d_1 \ell +\delta_\CC- z_0|},
\qquad \mbox{\rm for } \ell d_1 \in \cZBCC{(c)}. \quad (c=0, 1, 2).
$$

\begin{proposition}\label{prop:singleBCC}
The single screw dislocation around 
$\pi_{\EE}^{-1}(z_0) \subset \EE^3=\EE_{\EE_\CC}$ 
expressed by 
$$
\bigcup_{c=0}^2
\widehat\iota_{\delta_\CC} \left(
\widehat{\psi}_{d_0}^{-1}
\left(\omega_3^{-c} \displace^\BCC_{z_0,\delta}
(\cZBCC{(c)})\right)\right)
$$
 is a subset of $\EE^3$.
\end{proposition}

\begin{proof}
$\EE_{\EE\setminus\{z_0\}}$ is obviously a subset of $\EE^3=\EE_{\EE_\CC}$.
\end{proof}

\begin{remark}
{\rm{
Though it is obvious that the screw dislocation exists in $\EE^3$ in physics,
it is not obvious that a geometrical object constructed in
algebraic topology is realized in $\EE^3$, e.g, the Klein bottle. 
Proposition \ref{prop:singleBCC}
is crucial in the description of the physical object
in terms of algebraic language.
}}
\end{remark}

As in the SC lattice, we also consider the 
graph $G_{z_0,\delta}^\BCC$.
In the following, we also assume that $\gamma = 1$
and $\delta_\CC=0$ for simplicity.

\subsection{Note on the Core Region of Screw Dislocations in BCC Lattice}
\label{subsection4.5}

Though the core structure in the screw dislocations in the BCC lattices
has been studied well, e.g., in \cite{Tak, ARO, ATMO} and using the 
first principle approach \cite{Cl, GD, IKY},
we show the description of the core region in terms of our framework.

Noting the equations
(\ref{eq:ascendent}) and 
(\ref{eq:descendent}), we consider the core region of the screw dislocation.
The core region is the cells neighborhood of $z_0$.

If $z_0$ is the center of the ascendant triangle, 
for a vertex $z$ of the triangle,
the set of the fiber direction is
$$
D(z)=\psi_{d_0}^{-1}\left(\gamma_{\delta}\frac{(z-z_c)^2}{|z-z_c|^2}\right)
=\psi_{d_0}^{-1}\left(\gamma_{\delta}\frac{z-z_c}{\overline{z-z_c}}\right).
$$
On the other hand,
if $z_0$ is the center of the decedent triangle,
the set of the fiber direction is
$$
D(z)=\psi_{d_0}^{-1}\left(\gamma_\delta\right)
$$
for each vertex $z$ of the triangle.
They are illustrated in Figure \ref{fig:BCC05} (c) and (d) respectively.
In the former case, there might exist different connections illustrated by
the dotted lines.
Thus the screw dislocation in the BCC lattice
 shows the quite different aspect from the case of the SC lattice.

The operation in Remark \ref{rk:2.2}
 can be applied to this system
so that we have Figure \ref{fig:BCC05} (c) and (d).
By the operation, the connected spiral paths are deformed
 the disjointed paths.
The disjoint subgraphs characterize the direction of the screw dislocations.

\subsection{Energy of Screw Dislocation in BCC Lattice }\label{section6.1}

In this section, we estimate the stress energy of the screw dislocation
in the BCC lattice 
in the meso-scopic scale.
We basically investigate the energy in parallel with the computations
 in the SC lattice.

For simply convention, we denote $\displace^\BCC_{z_0,\delta}$
etc.\ simply by $\displace^\BCC_{z_0}$ etc.\ 
by suppressing $\delta$.

For $\ell\in \ZZ[\omega_6] +\mu_c$  $(c=0,1,2)$, we define
the relative height differences
$\varepsilon_{\ell}^{(c,j)}$ $(j=0,1,\cdots,5)$,
\begin{equation}
\begin{split}
\varepsilon_{\ell}^{(c,j)}
&= \frac{d_0}{2\pi\ii}
    \left(\log(\displace^\BCC_{z_0}((\ell+\nu_{j})d_1)
    -\log(\displace^\BCC_{z_0}(\ell d_1)) \right)\\
&= \frac{d_0}{4\pi\ii}
  \left(
\log\left(1 + \frac{d_1\nu_{j}}{\ell d_1-z_0}\right) 
-\log\left(1 + \frac{\overline{d_1\nu_{j}}}{\overline{\ell d_1-z_0}}
\right) \right).
\end{split}
\end{equation}
Here we require that
$-d_3/2 < \varepsilon_{\ell}^{(c,j)} < d_3/2$.
Let us introduce a parameter $\varepsilon >0$ and
using it,  we define the core region $C^{\BCC (j)}_{\varepsilon,I}$ of 
type I,
$$
C^{\BCC (c)}_{\varepsilon,\mathrm{I}}:=
\{\ell \in \ZZ[\omega_6] +\mu_c\ | \ {}^\exists{j} = 0,1,\cdots,5 
\mbox{ such that } 
|\varepsilon_{\ell}^{(c,j)}| > \varepsilon \}.
$$
Assume that $\varepsilon < d_3/2$.
The difference of length in each segment between
 $\ell \in \cZBCC{(c)}\setminus d_1C^{\BCC (c)}_{\varepsilon,\mathrm{I}} $ 
and its nearest neighbor lattice points
is given by 
$$
\Delta_{\ell}^{(c,j)} =
\sqrt{\left(d_3+(-1)^j\varepsilon_{\ell}^{(c,j)}\right)^2 +d_2^2}
-\sqrt{d_3^2+d_2^2},
$$
for
$j=0,1,\cdots,5$. Here we note that $\sqrt{d_3^2+d_2^2}=\sqrt{3}a/2=d_0$.

We have the following.
\begin{lemma} \label{lm:elasap1}
If 
$\displaystyle{\frac{d_1}{\sqrt{|\ell d_1-z_0|^2}}}$
for an $\ell\in \ZZ[\omega_6] +\mu_c$  $(c=0,1,2)$
is sufficiently small, 
$\varepsilon^{(c,j)}_{\ell}$'s
are approximated by
\begin{equation}
    \varepsilon_{\ell}^{(c,j)}  =  
\frac{d_0d_1}{4\pi\ii}
\left(
    \frac{\nu_{j}}{\ell d_1-z_0}
-\frac{\overline{\nu_{j}}}{\overline{(\ell d_1-z_0)}}
\right)
 + o\left(
\frac{d_1}{\sqrt{|\ell d_1 -z_0|^2}}
\right), 
\end{equation}
respectively,
whereas $\Delta_{\ell}^{(c,j)}$
are approximated by 
\begin{equation}\label{eq:approx_of_delta2}
    \displaystyle{\Delta_{\ell}^{(c,j)} 
= \frac{(-1)^jd_3}{d_0}
    \varepsilon_{\ell}^{(c,j)}
 + o\left( \frac{d_1}{\sqrt{|\ell d_1 -z_0|^2}} \right).}
\end{equation}
\end{lemma}

\begin{proof}
Noting $\log (1+z) = z +o(z^2)$,
$\displaystyle{
\sqrt{1+z}-1 = \frac{1}{2}z +o(z^2)}
$,
$\sqrt{d_3^2+d_2^3}=\sqrt{3}a/2=d_0$,
the direct computations show them.
\end{proof}

As mentioned in Remark \ref{rmk:AlgBCC},
due to the properties of the Eisenstein integers
$\ZZ[\omega_g]$
in Lemmas \ref{lm:cyl} and \ref{lm:4.2}, we have
the simple expression:
\begin{lemma} \label{lm:elasap2}
If 
$\displaystyle{\frac{d_1}{\sqrt{|\ell d_1-z_0|^2}}}$
for an $\ell\in \ZZ[\omega_6] +\mu_c$  $(c=0,1,2)$
is sufficiently small, 
\begin{equation}\label{eq:approx_of_delta3}
\frac{1}{2}\sum_{j=0}^5(\Delta_{\ell}^{(c,j)})^2 =
\frac{1}{384\pi^2}
 \frac{d_1^4}{|\ell d_1-z_0|^2}
 + o\left( \frac{d_1}{\sqrt{|\ell d_1 -z_0|^2}^3} \right).
\end{equation}
\end{lemma}

\begin{proof}
The left hand side is equal to
$$
-\frac{1}{2}\frac{d_0^2d_1^2d_3^2}{16\pi^2d_0^2}
\sum_{j=0}^5 
\left(
    \frac{\nu_{j}}{\ell d_1-z_0}
-\frac{\overline{\nu_{j}}}{\overline{\ell d_1-z_0}}
\right)^2
 + o\left( \frac{d_1}{\sqrt{|\ell d_1 -z_0|^2}^3} \right),
$$
and thus using Lemmas \ref{lm:cyl} and \ref{lm:4.2}, it becomes
$$
=
\frac{d_1^4}{384\pi^2}
\frac{1}{|\ell d_1-z_0|^2}
 + o\left( \frac{d_1}{\sqrt{|\ell d_1 -z_0|^2}^3} \right).
$$
Here the extra terms are canceled due to the properties of 
$\frac{1}{3}\ZZ[\omega_6]$.
\end{proof}

\begin{remark}\label{rmk:fC6action1}
{\rm{
As we show in Remark \ref{rmk:epsilon_w},
$\varepsilon_{\ell}^{(c,j)}$ is a real analytic
function of $\displaystyle{w_j:= \frac{d_1\nu_{j}}{\ell d_1-z_0}}$
and $\overline{w}_j$
for $|w_j|\ll 1$, i.e.,
$$
\varepsilon_{\ell}^{(c,j)}=
\varepsilon_{\ell}(w_j, \overline{w_j})  =  
\frac{d_0}{4\pi\ii}\log
\left(\frac{1+w_j}{\overline{1+w_j}}\right),
$$
and thus
$\Delta_{\ell}^{(c,j)}$ is also a real analytic function 
of $w$ and $\overline{w}$.

As we mentioned in Remark \ref{rmk:fC6action0},
we have the cyclotomic symmetry in $\ZZ[\omega_6]$, or
$\fC_6$ action on $\ZZ[\omega_6]$;
there is the element $g \in \fC_6$ such that 
$g \nu_{j} = \nu_{j+1}$ whose index is given modulo 6.
The action induces the action on the function $f$ over $\ZZ[\omega_6]$,
i.e., $g^* f(x) = f(g x)$.
By fixing $\ell d_1 - z_0$, the action is given as
$$
g^* 
\varepsilon_{\ell}^{(c,j)}=
g^*\varepsilon_{\ell}(w_j, \overline{w_j}) =
\varepsilon_{\ell}^{(c,j+1)},
$$
and $\displaystyle{
\sum_{j=0}^5 
(\Delta_{\ell}^{(c,j)})^2
}$ in Lemma \ref{lm:elasap2}
is invariant for the action $g \in \fC_6$.
}}
\end{remark}

For  a positive number $\rho$,
let us define another core region $C^{\BCC (c)}_{\rho,\mathrm{II}}$ 
of type II,
$$
C^{\BCC (c)}_{\rho,\mathrm{II}}:=
\left\{\ell \in \ZZ[\omega_6] +\mu_c \ | \ 
\left|\ell d_1-z_0\right| < \rho d_1 \right\}.
$$
In order to avoid to count doubly,
we should concentrate one of $\cZBCC{(c)}$'s and choose
$\cZBCC{(0)}$ in this paper.
Let the core region of type III and its compliments be 
$$
C^{\BCC (0)}_{\varepsilon,\rho,\mathrm{III}}
:= 
\{\ell \in \ZZ[\omega_6]  \ | \ \ell \in  
C^{\BCC (0)}_{\varepsilon,\mathrm{I}}
\cup C^{\BCC (0)}_{\rho,\mathrm{II}} \mbox{ or }
\mathrm{Ad}(\ell d_1) \subset \bigcup_{c=1}^2 
C^{\BCC (c)}_{\varepsilon,\mathrm{I}} \}
$$
and
$$
A^{\omega_6}_{\varepsilon,\rho}:=
\ZZ[\omega_6] \setminus C^{\BCC (0)}_{\varepsilon,\rho,\mathrm{III}},
\qquad
A^{\omega_6}_{\varepsilon,\rho,N}:=
\{\ell \in A^{\omega_6}_{\varepsilon,\rho} \ | \ 
|\mathrm{dist}_{z_0} (\mathrm{Ad}(\ell d_1)| <N d_1\}.
$$
where 
for a node $v$ in the plane graph $\pi_G(G_{z_0}^\BCC)$, we denote the set of
the adjacent nodes of $v$ by $\mathrm{Ad}(v)$ and 
we define 
$$
\mathrm{dist}_{z_0}(\{v_i\}) := \max_{v\in\{v_i\}}  |v -z_0|.
$$

For the case that $\rho$ is sufficiently large for given $\varepsilon
\displaystyle{\left(<\frac{d_3}{2}\right)}$ so that
$C^{\BCC (0)}_{\varepsilon,\rho,\mathrm{III}} 
\subset C^{\BCC (0)}_{\rho,\mathrm{II}}$,
we also define
$$
A^{\omega_6}_{\rho}:=
A^{\omega_6}_{\varepsilon,\rho}, \quad
A^{\omega_6}_{\rho,N}:=
A^{\omega_6}_{\varepsilon,\rho,N}.
$$

We compute the stress energy caused
by the screw dislocation in the BCC lattice as in the SC lattice case.
We compute the energy density for unit length in the
$(1,1,1)$-direction, and call it simply the stress
energy of dislocation again.

Let $k_d$ be the spring constant of the edges.
The stress
energy of dislocation in the 
annulus region
$A^{\omega_6}_{\varepsilon,\rho,N}$ is given by
\begin{equation}
E^\BCC_{\varepsilon,\rho,N}(z_0) := \sum_{\ell \in A^{\omega_6}_{\varepsilon,\rho,N}}
\cE^\BCC_{\ell},
\label{eq15aaB}
\end{equation}
where 
$\cE^\BCC_{\ell}$ for every $\ell\in \ZZ[\omega_6]$
is the energy density
defined by 
\begin{equation}
\cE^\BCC_{\ell} := 
\frac{1}{2}k_d
\sum_{j=0}^5 \left(\Delta_{\ell}^{(j)}\right)^2.
\end{equation}

As in Proposition \ref{prop:9} in the SC lattice, we summarize 
the above results as the stress energy of the
dislocation of the BCC lattice case:

\begin{proposition} \label{prop:9B}
\begin{itemize}
\item[$(1)$] For $\ell \in A^{\omega_6}_{\varepsilon,\rho,N}$,
the energy density $\cE_{\ell}$ is
expressed by a real analytic function $\cE^\BCC(w, \overline{w})$ of
$w$ and $\bar{w} \in \CC$ with $|w| < 1/\sqrt{2}$ in such a way that
\begin{equation*}
\cE^\BCC_{\ell} = \cE^\BCC\left(
\frac{d_1}{\ell d_1 - z_0},
\frac{d_1}{\overline{\ell d_1 - z_0}}
\right).
\end{equation*}
\item[$(2)$] Let us consider the power series expansion
$$
\cE^\BCC(w,\overline w) = \sum_{s=0}^\infty
\cE_\BCC^{(s)}(w, \overline w), \quad
\cE_\BCC^{(s)}(w, \overline w) := \sum_{i+j=s, i,j\ge0} C_{i, j} 
w^i \overline{w}^j,
$$
for some $C_{i, j} \in \CC$.
Then, we have the following:
\begin{itemize}
\item[\textup{(a)}]  
$\cE_\BCC^{(0)}(w, \overline w)=\cE_\BCC^{(1)}(w, \overline w)=0$,
\item[\textup{(b)}]  the leading term is given by 
\begin{equation}
\cE_\BCC^{(2)}(w, \overline w) =
\frac{d_1^4}{384\pi^2} k_d 
w\overline{w},
\qquad
\cE_\BCC^{(2)}\left(
\frac{d_1}{\ell d_1 - z_0},
\frac{d_1}{\overline{\ell d_1 - z_0}}
\right)=
\frac{1}{384\pi^2} k_d
 \left[
  \frac{d_1^4}{|\ell d_1 - z_0|^2 }
 \right], 
\label{eq15aB}
\end{equation}
\item[\textup{(c)}]  $C_{i, j}=\overline{C_{j, i}}$, and
\item[\textup{(d)}]  for every $s \geq 2$, there is a constant $M_s > 0$ such that
$$
|\cE_\BCC^{(s)}(w, \overline w)| \le M_s |w|^s. 
$$
\end{itemize}
\end{itemize}
\end{proposition}

\begin{proof}
Remark \ref{rmk:fC6action1}
Lemmas \ref{lm:elasap1} and \ref{lm:elasap2}
show (1) and (2) (a), (b).
Since the energy density is a real number, we obtain the
relation in item (c). The analyticity in item (1) implies (d).
This completes the proof.
\end{proof}

\bigskip

As the summation in (\ref{eq15aaB}) is finite, we have
\begin{equation}
E^\BCC_{\varepsilon,\rho,N}(z_0) =  
\sum_{s=2}^\infty
\sum_{\ell\in A^{\omega_6}_{\varepsilon,\rho,N}}
\cE_\BCC^{(s)}\left(
\frac{d_1}{\ell d_1 - z_0},
\frac{d_1}{\overline{\ell d_1 - z_0}}
\right).
\label{eq-seriesB}
\end{equation}

In particular, we have the following theorem for the 
``principal part'' of the stress energy.
\begin{theorem}\label{thm:energyBCC}
For the case that $\rho$ is sufficiently large for given $\varepsilon
\displaystyle{\left(<\frac{d_3}{2}\right)}$ so that
$C^{\BCC (0)}_{\varepsilon,\rho,\mathrm{III}} \subset
 C^{\BCC (0)}_{\rho,\mathrm{II}}$.
Let $A^{\omega_6}_{\rho,N} :=A^{\omega_6}_{\varepsilon,\rho,N}$.
The principal part of the stress energy $E_{\rho,N}(z_0)$, defined by
\begin{eqnarray*}
  E^{\BCC(\mathrm{p})}_{\rho,N}(z_0) &:=&
   \sum_{\ell \in A^{\omega_6}_{\rho,N}} 
\cE_\BCC^{(2)}\left(
\frac{d_1}{\ell d_1 - z_0},\frac{d_1}{\overline{\ell a - z_0}}
\right)  \\ 
&=& \frac{1}{384\pi^2} k_d
   \sum_{\ell \in A^{\omega_6}_{\rho,N}}
   \left[
    \frac{d_1^4}{|\ell d_1 -z_0|^2}
   \right],
\end{eqnarray*}
is given by  the truncated
Epstein-Hurwitz zeta function (see Appendix),
\begin{equation}
E^{\BCC(\mathrm{p})}_{\rho,N}(z_0)
= \frac{1}{384\pi^2}k_d d_1^2
\zeta_{\rho, N}^{\omega_6}(2, -z_0/d_1).
\label{eq:th1B}
\end{equation}
\end{theorem}

It is noted that 
this theorem is obtained due to the properties of 
$\displaystyle{\frac{1}{3}\ZZ[\omega_6]}$ cf.
Remark \ref{rmk:AlgBCC}.

By Proposition~\ref{prop:9B} (2) (d),
we can estimate each of the other terms appearing
in the power series expansion (\ref{eq-seriesB})
by the truncated
Epstein-Hurwitz zeta function for Eisenstein 
integers as follows.

\begin{proposition}\label{prop:remainderB}
For each $s \geq 3$, there exists a positive constant $M_s'$ such that 
\begin{equation}
 \sum_{\ell\in A^{\omega_6}_{\rho,N}}
\cE_\BCC^{(s)}\left(
\frac{d_1}{\ell d_1 - z_0}, \frac{d_1}{\overline{\ell d_1 - z_0}}
\right)  \le M_s' \zeta_{A_{\rho, N}^{\omega_6}}^{\omega_6}(s, -z_0/d_1).
\end{equation}
\end{proposition}

\section{Discussion}
In this paper, we investigated  the screw
dislocations of the SC lattice and the BCC lattice
using the number theoretic descriptions 
in terms of the Gauss integers $\ZZ[\ii]$
and the Eisenstein integers $\ZZ[\omega_6]$.

As mentioned in Remark \ref{rmk:AlgBCC},
using the properties of the Eisenstein integers $\ZZ[\omega_6]$, e.g.,
Lemmas \ref{lm:cyl} and \ref{lm:4.2},
we obtain the simple description of 
the stress energy for the screw dislocation for the finite region except
the core region.
It reflects the symmetry of the screw dislocations.
It is quite natural to investigate the symmetry of a mathematical
object using algebraic language.
Without the representation of the dislocation in terms of the 
Eisenstein integers $\ZZ[\omega_6]$, it is very
 difficult to obtain the result
because the dislocation in the BCC lattice  is very complicate. 
Even for the core region, we can investigate it 
as in Subsection \ref{subsection4.5} using 
the properties of $\ZZ[\omega_6]$.
Our description is natural even when we consider the analytic property
like the energy minimum point of the screw dislocations
because the symmetry is built in the descriptions.

We could estimate the dislocation of meso-scopic scale
because the stress energy $E_{\mathrm{total}}$ is given by
$$
E_{\mathrm{total}}=
E_{\mathrm{core}}+
E_{\mathrm{meso}}.
$$
The effect in the core region should be investigated by the 
first principle computations but the meso-scopic energy
could not be obtained.
Even though we need more precise investigation for the
estimation because there are some parameters, 
we have the formula to evaluate the meso-scopic energy.
It is noted that the core energy is determined by the local data
whereas the meso-scopic energy is determined by the meso-scopic data.

\begin{figure}[t]
\begin{center}
\includegraphics[width=13cm]{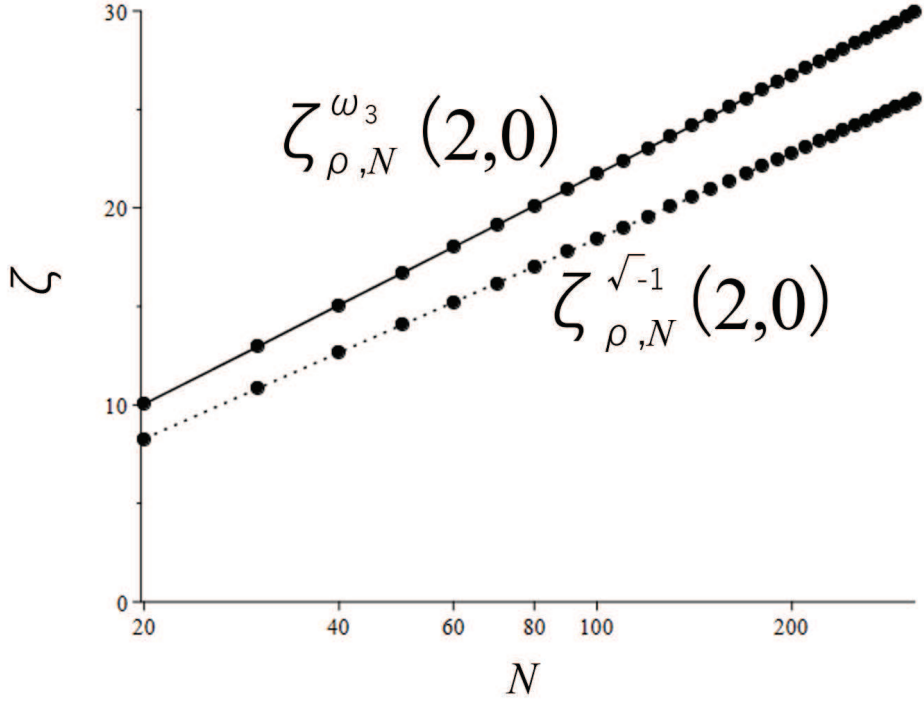}
\end{center}
\caption{
The graph of $\zeta_{\rho,N}^\tau(2,0)$ v.s. $\log N$ for $\rho=5.1$}
\label{fig:graph_zeta}
\end{figure}

The energy of the meso-scale essentially diverges 
and thus it is important to determine the cut-off parameter $N$.
Let $\zeta_{\rho, N}^\tau(s, z_0) :=\zeta_{A_{\rho, N}^\tau}^\tau(s, z_0) $.
As we show the behaviors of the $\zeta_{\rho,N}^\tau(2,0)$ for 
$\tau = \ii$ and $\omega_6$ in  
Figure \ref{fig:graph_zeta},
they are approximated well by the logarithmic function.
It is natural since the continuum theory, in which the dislocation
energy $E_{\mathrm{total}}(R)$ in the inner region $\{z \in \CC\ | |z-z_0|<R\}$
is written by the logarithmic function with respect to the
radius from the dislocation line;
$E_{\mathrm{total}}(R) \propto \log R$.

Further we show the density of the $\zeta_{\rho,N}^\tau(2,x+y\tau')$
as in Figure \ref{fig:zeta_density} by numerical computations;
the region of $z_0/d$ is divided by $20 \times 20$ blocks.
These aspects in the regions are different though the differences are not large
due to the divergent properties like the logarithmic function.

\begin{figure}[t]
\begin{center}
\includegraphics[width=13cm]{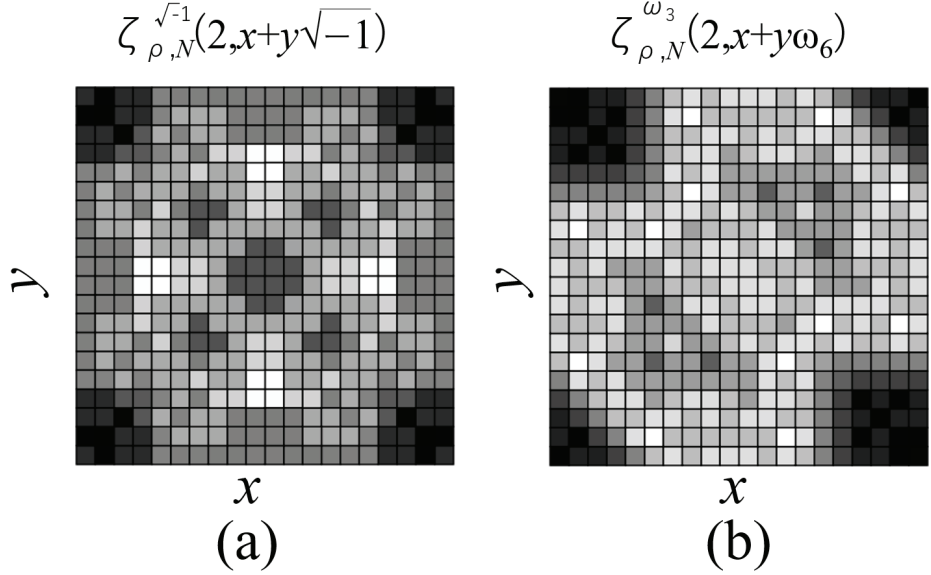}
\end{center}
\caption{
The graph of $\zeta_{\rho,N}^\tau(2,x+y\tau')$ for. $(\rho,N)=(7.2,75)$:
(a) $\zeta_{\rho,N}^\ii(2,x+y\ii)$ with gray scale: black = 14.664, white = 14.779
and (b) 
$\zeta_{\rho,N}^{\omega_6}(2,x+y\omega_6)$ with gray scale: black =16.061, white = 16.907}
\label{fig:zeta_density}
\end{figure}

As we computed the double dislocations case in the SC lattice in the previous
work \cite[Appendix]{HMNSU}, they are described well by the Green function
in the statistical field theory like vortexes as in \cite{AO,ID}.
In the computations of the Green function, there appear 
the quadratic form $k \ell \in \CC$ modulo $2\pi \ZZ$, where
$k \in \QQ(\tau) \pi/d$ and 
$\ell \in \ZZ(\tau)d$ for $d=a$ or $d=d_1$. 
These computations are very crucial in the quadratic number theory \cite{Tri}.
If the distance between the dislocation is larger enough, the behavior of
the dislocations are determined by the continuum theory.
However otherwise, 
it implies that the prime numbers in the Gauss integer or the Eisenstein
integer (Gauss primes or Eisenstein primes) might have effects on the
configurations of the dislocations if the meso-scopic energy
plays crucial in the total energy.

\section{Conclusion}

In this paper, we show a discrete investigation on the screw dislocations
of the SC lattices and the BCC lattices in terms of the 
elementary  number theory.
It is well-known that the two-dimensional
lattices in the SC lattice perpendicular to $(0,0,1)$-direction,
and in the BCC lattices perpendicular to $(1,1,1)$-direction 
are described in terms of the Gauss integers $\ZZ[\ii]$
and the Eisenstein integers $\ZZ[\omega_6]$ respectively. 
Since the Burger vectors in the screw-dislocation are the
$(0,0,1)$ and the $(1,1,1)$ directions for 
the SC and the BCC lattices respectively,
we use the facts and show the following:
\begin{enumerate}
\item 
The displacement caused by the screw dislocations are expressed by 
the functions of the Gauss integers $\ZZ[\ii]$ in the SC lattice case,
and of the Eisenstein integers $\ZZ[\omega_6]$ 
in the BCC lattice case respectively, as in
Propositions \ref{prop:SC single} and \ref{prop:singleBCC}.
\item 
As
mentioned in Remarks \ref{rmk:AlgSC} and \ref{rmk:fC6action0},
the cyclic groups $\fC_4$ and $\fC_6$ act on the dislocations
via $\ZZ[\ii]$ and $\ZZ[\omega_6]$ respectively,
and the action plays important roles in the evaluation of
the stress energy of the dislocations in these lattices as in
Remarks  \ref{rmk:epsilon_w} and \ref{rmk:fC6action1}.
\item Due to the symmetry, we explicitly
evaluate the stress energy density
of the screw dislocations in the meso-scopic scale
in terms of the truncated Epstein-Hurwitz zeta functions 
as in Propositions \ref{prop:9} and 
\ref{prop:9B}.
We remark that
without number theoretic descriptions, it is difficult to
obtain these expressions of the energy in terms of 
the truncated Epstein-Hurwitz zeta functions.
\item With knowledge in the elementary number theory,
we can explicitly 
evaluate the leading term of the stress energy density
in the meso-scopic scale
using the truncated Epstein-Hurwitz
zeta functions as in Theorems \ref{thm:energySC} and 
\ref{thm:energyBCC}, and Figures \ref{fig:graph_zeta}
and \ref{fig:zeta_density} in Discussion.
They show the meso-scopic contributions of the 
stress density in the screw dislocations.
\end{enumerate}

As mentioned in Introduction, since the crystal lattices with dislocations 
even have high symmetries, 
we should investigate the dislocations by considering the 
symmetries. The number theoretic approach is a practical tool to describe their
symmetries, translations and rotations. We demonstrated that the number theoretic approach reveals the properties of the dislocations and recovers the 
stress energy in 
the continuum picture using the Epstein-Hurwitz zeta functions.
Further since the Gauss integers, the Eisenstein integers and 
the Epstein-Hurwitz zeta functions have interesting properties,
we might find more crucial phenomenon by cooperating with 
analytic considerations in the future.
Thus we expect that our method shed light on novel investigations on
dislocations.

\section*{Acknowledgments}
The author thanks to all those who participated in the problem session
``Mathematical description of disordered structures in crystal''
in the Study Group Workshop 2015 held in Kyushu University
and in the University of Tokyo during July 29--August 4, 2015,
and to the participants in the
``IMI workshop II: Mathematics of Screw Dislocation'',
September 1--2, 2016, in the
``IMI workshop I: Mathematics in Interface, Dislocation and 
Structure of Crystals'',
August 28--30, 2017, 
``IMI workshop I: Advanced Mathematical Investigation of Screw Dislocation'',
September 10--11, 2018"
held in Institute of Mathematics for Industry (IMI), in 
Kyushu University, especially Shun-ichi Amari, Toshikazu Sunada,
Tetsuji Tokihiro, Kenji Higashida, Hiroyuki
Ochiai, and Kazutoshi  Inoue for variable discussions and comments.
He is also grateful to the authors in \cite{HMNSU}, Hiroyasu Hamada, Junichi 
Nakagawa, Osamu Saeki and Masaaki Uesaka for helpful discussions and comments.
The author has been supported by
JSPS KAKENHI Grant Number 15K13438 and by Takahashi Industrial and 
Economic Research Foundation 2018-2019, 08-003-181.
He also thanks to anonymous referees for critical and helpful comments.
\appendix

\section{Two-dimensional lattices and the Epstein-Hurwitz zeta function}

When we regard two dimensional lattice $L_{(a_1,a_2)}$ 
as the free $\ZZ$-modules, 
$
L_{(a_1,a_2)} = \ZZ a_1 + \ZZ a_2 (\subset \RR^2),
$
for unit vectors
 $a_1, a_2\in \CC$, where $a_1$ and $a_2$ are linear independent.
It is obvious that the lattice has the unit cell.
When we consider the classification of $L_{(a_1,a_2)}$,
or its moduli space (its parameter space), 
it is natural to introduce the normalized lattice
$
L_{\tau} = \ZZ + \ZZ \tau,
$
for $(1, \tau :=a_2/a_1)$. We assume  
$\tau \in \HH:=\{x+\ii y\in \CC\ |\ y > 0\}$ without loss of generality.
However there are ambiguities which ones are regarded as the
unit vectors. 
There is an action of $\SL(2, \ZZ)$ as an automorphism on  
 $L_{\tau}\times L_{\tau}$;
 for $(\ell_1, \ell_2)$ and
${g:=\begin{pmatrix} a & b\\ c& d\end{pmatrix} \in \SL(2, \ZZ)}$
$\displaystyle{
g (\ell_1, \ell_2)= {}^t(g \ {}^t(\ell_1, \ell_2))
=(a\ell_1+b \ell_2 \tau, a\ell_1+d \ell_2 \tau)}$ so that the 
area of the parallelogram generated  by $\ell_1$ and $\ell_2$ preserves.
Here  for the parallelogram generated by 
$z_1=x_1+y_1\ii$ and $z_2=x_2+y_2\ii$, its area is equal to 
$x_1 y_2- x_2 y_1$.
Thus we regard every element $g (1,\tau)$ in $\SL(2, \ZZ)(1,\tau)=
\{g (1,\tau)\ | \ g \in \SL(2, \ZZ)\}$ as the unit vector in 
$L_{\tau}$.

Therefore the M\"obius transformation (for $g\in \SL(2, \ZZ)$,
$g (z_1:z_2):=(a z_1 + b z_2: c z_1 + d z_2)$) is also introduced,
which is denoted by $\PSL(2, \ZZ)$. By regarding $g\tau := g(1:\tau)$, 
it induces
a natural group action of $\PSL(2, \ZZ)$ on $\HH$.
The fundamental domain as the moduli of $L_{\tau}$ turns out to be
$\HH/\PSL(2, \ZZ)$.

The following are well-known facts: e.g.,\cite[Theorem 8.5]{Knapp}
\begin{lemma}
\label{lm:SL2Z}
For a point $\tau \in \HH/\PSL(2, \ZZ)$,
the stabilizer subgroup $G_\tau$ of 
$\SL(2,\ZZ)$, $G_\tau:=\{g \in \SL(2, \ZZ) | g (1,\tau) = (1,\tau)\}$,
becomes a cyclic group $\fC_n:=\{ t^\ell\ |\ \ell = 0, 1, \ldots, n-1, \ 
t^n= t^0\}$ of the order $n$, i.e.,
\begin{enumerate}
\item $\tau =\ii=\omega_4$, $G_\tau = \fC_4$

\item  $\tau = \omega_6$, $G_\tau = \fC_6$, and

\item otherwise, $G_\tau = \fC_2$,
\end{enumerate}
where $\omega_p:=\ee^{2\pi\ii/p}$.
\end{lemma}

In this paper, the both $\omega_4$ and $\omega_6$ play
the crucial role. 
Let $\ZZ[\tau]:=\{\ell_1 + \ell_2 \tau\ | \ \ell_1, \ell_2 \in \ZZ\}$
as a discrete subset of $\RR^2$ and $\CC$.
The set of the Gauss integers is denoted by
$\ZZ[\ii]$ and 
the set of the Eisenstein integers is by
$\ZZ[\omega_6]=\ZZ[\omega_3]$ for $\omega_3=\omega_6^2$
noting  $\omega_3+1=\omega_6$.

The \textit{truncated Epstein-Hurwitz zeta function}  of  
 $\tau \in \HH$, $\zeta_{A}^\tau(s, z_0)$ is defined by \cite{Tre}
\begin{equation}
    \zeta_{A}^{\tau}(s, z_0) := \sum_{\ell\in A}
    \frac{1}{(|\ell+z_0|^2)^{s/2}},
    \label{eq:tzetaB}
\end{equation}
where $z_0:=x_0+y_0 \ii \in \CC$ and $A$ is a subset of 
$\ZZ[\tau]$.

%\bigskip

\noindent
Shigeki Matsutani\\
Graduate School of Natural Science and Technology,\\
Kanazawa University\\
Kakuma Kanazawa, 920-1192, JAPAN\\
\textit{s-matsutani@se.kanazawa-u.ac.jp}

\end{document}